\documentclass[pra,twocolumn,a4paper,showpacs,aps,10pt,nofootinbib,allowtoday, accepted=2022-10-06]{quantumarticle}
\pdfoutput=1
\usepackage{graphicx}
\usepackage{hyperref}
\usepackage[section]{placeins}
\usepackage{amsmath}
\usepackage{amsthm}
\usepackage{amssymb}
\usepackage{dsfont}
\usepackage{bbold}
\usepackage{color}

\newtheorem{theorem}{Theorem}

\newcommand{\Tr}{\text{Tr}}

\newcommand{\ack}{\subsection*{\normalsize \sf \textbf{Acknowledgment}}}

\newcommand{\bit}{\begin{itemize}}
	\newcommand{\eit}{\end{itemize}\par\noindent}
\newcommand{\ben}{\begin{enumerate}}
	\newcommand{\een}{\end{enumerate}\par\noindent}
\newcommand{\beq}{\begin{equation}}
\newcommand{\eeq}{\end{equation}\par\noindent}
\newcommand{\beqa}{\begin{eqnarray*}}
	\newcommand{\eeqa}{\end{eqnarray*}\par\noindent}
\newcommand{\beqn}{\begin{eqnarray}}
\newcommand{\eeqn}{\end{eqnarray}\par\noindent}

\usepackage{hyperref}
\usepackage{amsmath}
\usepackage{amsthm}
\usepackage{graphicx}
\usepackage{braket}
\usepackage[makeroom]{cancel}

\newtheorem{proposition}{Proposition}

\begin{document}
	\title{Contextuality in Entanglement-assisted One-shot Classical Communication}
	
	\author{Shiv Akshar Yadavalli}
	\email{sy215@duke.edu}
	\affiliation{Department of Physics, Duke University, Durham, North Carolina, USA 27708}
	
	\author{Ravi Kunjwal}
	\email{ravi.kunjwal@ulb.be}
	\affiliation{Centre for Quantum Information and Communication, Ecole polytechnique de Bruxelles,
		CP 165, Universit\'e libre de Bruxelles, 1050 Brussels, Belgium}
	
	\date{\today}                                           
	
	\begin{abstract}
		We consider the problem of entanglement-assisted one-shot classical communication. In the zero-error regime, entanglement can increase the one-shot zero-error capacity of a family of classical channels following the strategy of \href{ https://journals.aps.org/prl/abstract/10.1103/PhysRevLett.104.230503}{Cubitt {\em et al.}, Phys.~Rev.~Lett.~104, 230503 (2010)}. This strategy uses the Kochen-Specker theorem which is applicable only to projective measurements. As such, in the regime of noisy states and/or measurements, this strategy cannot increase the capacity. To accommodate generically noisy situations, we examine the one-shot success probability of sending a fixed number of classical messages. We show that preparation contextuality powers the quantum advantage in this task, increasing the one-shot success probability beyond its classical maximum. Our treatment extends beyond Cubitt {\em et al.} and includes, for example, the experimentally implemented protocol of 
		\href{https://link.aps.org/doi/10.1103/PhysRevLett.106.110505}{Prevedel {\em et al.}, Phys. Rev. Lett. 106, 110505 (2011)}. We then show a mapping between this communication task and a corresponding nonlocal game. This mapping generalizes the connection with pseudotelepathy games previously noted in the zero-error case. Finally, after motivating a constraint we term {\em context-independent guessing}, we show that contextuality witnessed by noise-robust noncontextuality inequalities obtained in \href{https://doi.org/10.22331/q-2020-01-10-219}{R. Kunjwal, Quantum 4, 219 (2020)}, is sufficient for enhancing the one-shot success probability. This provides an operational meaning to these inequalities and the associated hypergraph invariant, the weighted max-predictability, introduced in \href{https://doi.org/10.22331/q-2019-09-09-184}{R. Kunjwal, Quantum 3, 184 (2019)}. Our results show that the task of entanglement-assisted one-shot classical communication provides a fertile playground to study the interplay of the Kochen-Specker theorem, Spekkens contextuality, and Bell nonlocality.
	\end{abstract}
	
	\maketitle
\tableofcontents
\section{Introduction}\label{sec1}
The problem of identifying the resources responsible for a quantum advantage over classical strategies in quantum information and computation is key to unlocking the potential of quantum technologies. Often, such resources are taken to be theory-dependent features like entanglement, coherence, incompatibility, or perhaps the exponential scaling of Hilbert space dimension with the number of quantum systems at hand. The nonclassicality witnessed by Bell violations \cite{Bell64,CHSH} makes it possible to identify a source of quantum advantage that can be assessed in a theory-independent fashion, relying only on empirical data rather than internal features of the theory that generated the data. In contrast to the case of Bell nonlocality, Kochen-Specker (KS) contextuality \cite{KS67}, a notion of nonclassicality mathematically similar to Bell nonlocality, hasn't been as widely adopted as a theory-independent witness of quantum advantage. This is despite the existence of theoretical results on its relevance for quantum information and computation \cite{RW04, BBT05, CLM10, HWV14}. One reason for this is that it isn't robust to noise, unlike Bell nonlocality, making its experimental testability a matter of controversy \cite{BK04, Winter14}

Recently, much work has been devoted to making contextuality a notion of nonclassicality that relies on empirical data without making assumptions about the representation of measurements (concerning, in particular, their sharpness \cite{Kunjwal19, Cabello17, CY14}) in the theory generating the data \cite{Spekkens05, MPK16}.\footnote{Experimental testability of generalized contextuality -- in particular, the need for tomographic completeness in order to verify operational equivalences -- has been addressed in several recent papers \cite{MPK16, PdRM19, MPR21} and we refer the interested reader to these for a discussion of such issues and how they are handled in the framework. Note, however, that if one assumes a quantum description of any experiment (as opposed to a general probabilistic theory -- GPT -- description), then verification of operational equivalences via tomography is straightforward as there is no ambiguity about the dimension of a well-characterized system's Hilbert space, hence about the minimum number of tomographically complete preparations/measurements needed to establish operational equivalences. This move -- assuming a quantum description -- is the natural one when considering applications of contextuality in quantum information as opposed to its foundational implications for quantum theory (which necessitates a broader framework like GPTs). On the other hand, even if one does assume a quantum description, Kochen-Specker contextuality cannot do justice to the case where this description involves non-projective measurements for reasons that have been extensively discussed elsewhere \cite{Spekkens05, KS15, KS18, Kunjwal19}.}
This noise-robust notion of contextuality due to Spekkens \cite{Spekkens05} has been shown to underlie several quantum information tasks such as parity-oblivious multiplexing, quantum random access codes, state discrimination, communication complexity, anomalous weak values, and state-dependent cloning \cite{SBK09, CKKS16, SS18, SC19, SHP19, KLP19, LS20}. These applications of Spekkens contextuality, though, have no counterpart in terms of KS-contextuality, leaving a gap in our understanding of how advantages from KS-contextuality can be turned into noise-robust advantages premised on Spekkens contextuality. This is in line with the spirit of Ref.~\cite{KS15}, where the first noise-robust noncontextuality inequality, inspired by the Kochen-Specker theorem, was derived. Since the approach to noise-robust noncontextuality inequalities generalizes the KS paradigm by removing restrictions like projective measurements \cite{KS15}, it behooves us to ask if, and in what precise form, the advantages that derive from KS-contextuality persist when one considers noise-robust contextuality \`a la Spekkens \cite{Spekkens05, Kunjwal16}. In this paper, we take the first steps in this research program using tools from previously proposed hypergraph frameworks \cite{Kunjwal19, Kunjwal20}.

We consider the problem of one-shot classical communication where it has been shown that, assisted by entanglement, KS-contextuality provides an increase in the one-shot zero-error capacity of classical channels based on the KS theorem \cite{KS67, CLM10}. We study a relaxation of this problem to one of enhancing the one-shot success probability of sending a fixed number of classical messages assisted by entanglement.
Previous work \cite{PLM11,HMS13} has studied the one-shot success probability in the case of classical channels (unrelated to the KS theorem \cite{KS67}) which do not admit an enhancement of zero-error channel capacity \`a la Cubitt {\em et al.}~\cite{CLM10}. In contrast, we here study the one-shot success probability for the general case that includes, in particular, channels which {\em do} admit an enhancement of the one-shot zero-error channel capacity. A schematic of the task of entanglement-assisted one-shot classical communication is outlined in Fig.~\ref{schematic}.
\begin{figure}[htb!]
	\includegraphics[scale=0.2]{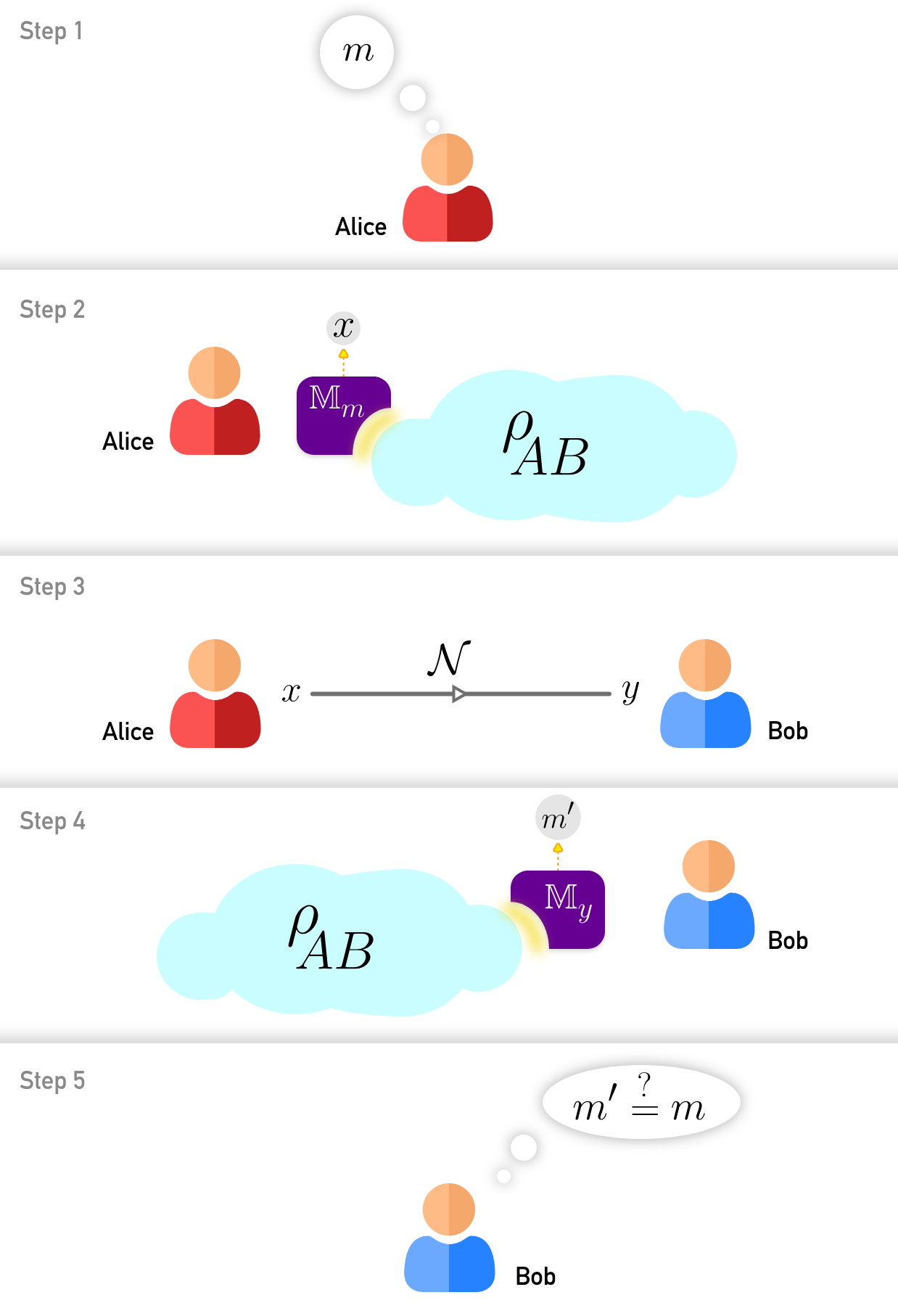}
	\caption{A schematic of the general protocol described in detail further on in Section \ref{sec3_1}. Alice and Bob are connected via a classical channel $\mathcal{N}$ and they share an entangled state $\rho_{AB}$. Once Alice decides to send a message $m$ in Step $1$, she encodes this message in her measurement choice in Step $2$ and obtains outcome $x$. In Step $3$, this outcome serves as the channel input  and Bob obtains channel output $y$ with probability $\mathcal{N}(y|x)$. Based on $y$, in Step $4$, Bob measures his quantum system and, after some classical post-processing, obtains his guess $m'$ for Alice's message $m$. They succeed if $m'=m$ (Step $5$).}
	\label{schematic}
\end{figure}

Our results can be summarized as follows:
\begin{itemize}
	\item  We show that preparation noncontextuality \cite{Spekkens05} relative to Bob's share of the system characterizes the classical upper bound on the task of enhancing the one-shot success probability of a classical channel assisted by entanglement.\footnote{This does not require that the channel be based on the KS theorem in the sense of Cubitt {\em et al.}~\cite{CLM10}.}
	Hence, preparation contextuality drives the quantum advantage in the general task. 
	
	\item We then show a mapping between the one-shot communication task and a corresponding nonlocal game: preparation contextuality powers an advantage in the first task if and only if Bell nonlocality powers an advantage in the second.\footnote{Note that while the first task, by definition, requires communication via a classical channel---hence timelike separation, where the channel output (to the receiver) is within the future lightcone of the channel input (from the sender)---the second task, by definition, forbids all communication by requiring spacelike separation, where the sender and receiver are causally disconnected during each run of the nonlocal game. However, in both situations---timelike or spacelike separation---the shared common-cause resource doesn't result in signalling, i.e., we assume that this resource is described by a non-signalling theory (\textit{e.g.}, shared entanglement in quantum theory) \cite{Barrett07}.}  This generalizes the connection between entanglement-assisted one-shot zero-error capacity and nonlocal (pseudotelepathy) games noted in Ref.~\cite{CLM10}.
	
	\item We motivate a constraint on the communication task that we term {\em context-independent guessing} in a situation where the receiver (Bob) has no knowledge of (or doesn't trust) the exact channel probabilities but knows only (or trusts only) the channel hypergraph. We then prove that, for some classical channels (including the one studied in Ref.~\cite{CLM10}), the contextuality witnessed by a hypergraph invariant -- the weighted max-predictability -- implies an enhancement of the one-shot success probability in the communication task. This makes direct connection with the formalism of Ref.~\cite{Kunjwal20}, where weighted max-predictability provides an upper bound on the strength of source-measurement correlations under the assumption of noncontextuality. We thus provide an operational meaning to the violation of noise-robust noncontextuality inequalities in Ref.~\cite{Kunjwal20}: namely, such violations power the enhancement of one-shot success probability of classical communication assisted by entanglement.
\end{itemize}
 
The structure of the paper is as follows: In Section \ref{sec2}, we define some preliminary notions from the theory of one-shot classical communication as well as contextuality. In Section \ref{sec3}, we describe the general protocol for entanglement-assisted one-shot classical communication that provides a unified description of protocols such as those of Refs.~\cite{CLM10} and \cite{PLM11}. In Section \ref{sec4}, we discuss the resources that play a role in the quantum advantage in the communication task, including preparation contextuality and Bell nonlocality. In Section \ref{sec5}, we dive deep into the connection between the role of preparation contextuality in our communication task and the role of Bell nonlocality in a corresponding nonlocal game, proving some general relationships between them. In Section \ref{sec6}, we look at the problem of one-shot communication of a single bit through classical channels with complete confusability graphs, in particular the classical channel of Ref.~\cite{PLM11}. In Section \ref{sec7}, we study the case of classical channels based on the KS theorem \`a la Ref.~\cite{CLM10}. We conclude with a discussion in Section \ref{sec8}, mentioning some open problems and opportunities for future work.

\section{Preliminaries}\label{sec2}

\subsection{Classical Channels}\label{sec2_1}
Consider a discrete and memoryless classical channel $\mathcal{N}$ (\textit{e.g.}, Fig.~\ref{prevedelchannelhypergraph}). Let $X$ denote the set of input symbols of $\mathcal{N}$ and $Y$ denote the set of output symbols so that  $\{\mathcal{N}(y|x)\}_{x\in X,y\in Y}$ denotes the channel probabilities satisfying: $\mathcal{N}(y|x)\geq 0$ for all $x\in X, y\in Y$ and $\sum_{y\in Y}\mathcal{N}(y|x)=1$ for all $x\in X$. Further, we denote by $Y_x\subseteq Y$ the set of output symbols that have a non-zero probability of occurrence when the input symbol is $x\in X$, i.e., the {\em support} of $x$, given by $Y_x\equiv \{y\in Y|\mathcal{N}(y|x)>0\}$ for $x\in X$.
Similarly, $X_y\subseteq X$ denotes the set of input symbols that yield a non-zero probability of occurrence for the output symbol $y\in Y$, i.e., the {\em support} of $y$, given by $X_y\equiv \{x\in X|\mathcal{N}(y|x)>0\}$ for all $y\in Y$.

To the classical channel $\mathcal{N}$, we associate the channel hypergraph $H(\mathcal{N})$: vertices of $H(\mathcal{N})$ denote the input symbols $x\in X$ and hyperedges denote the output symbols $y\in Y$, such that each hyperedge representing $y\in Y$ contains the input symbols in $X_y$. Any two input symbols $x,x'\in X$ are said to be {\em confusable} when they share a hyperedge in $H(\mathcal{N})$, i.e., $Y_x\cap Y_{x'}\neq\varnothing$. The {\em confusability graph} $G(\mathcal{N})$ of the channel is given by the orthogonality graph of $H(\mathcal{N})$, i.e., its vertices are given by $X$ and any two vertices in $X$ are connected by an edge if and only if they are confusable.

Given the classical channel $\mathcal{N}$, Alice and Bob choose an {\em encoding} of the messages (say, $[q]\equiv \{m\}_{m=1}^q$) that Alice (the sender) wants to send to Bob (the receiver) through the channel. An {\em encoding} is a collection of mutually disjoint subsets of $X$.

More concretely, $\{X^{(m)}\}_{m=1}^q$ (where $X^{(m)}\subseteq X$ for all $m\in[q]$) is an encoding of the set of $q$ messages in $[q]$ if and only if $X^{(m)} \cap X^{(m')}=\varnothing$ for all distinct $m,m’\in [q]$.\footnote{Hence, given any input symbol $x^{(m)}\in X^{(m)}$, the message $m$ can be uniquely inferred from $x^{(m)}$ since $X^{(m)}\cap X^{(m')}=\varnothing$ for all $m'\neq m$. In the interest of efficiency of encoding, we shall only consider encodings where each $X^{(m)}$ $(m \in [q])$ is a clique in $G(\mathcal{N})$.\footnote{This means that every pair of input symbols in a subset is confusable. An encoding which doesn't satisfy this property would be inefficient: it would fail to fully exploit the non-confusability properties of the channel in sending messages.}
}

A {\em zero-error code} is a set of input symbols $X_0\subseteq X$ that are mutually non-confusable, i.e., no two symbols in this set can map to the same output symbol when fed into the channel $\mathcal{N}$. Hence, an encoding $\{X^{(m)}\}_{m=1}^q$ of the messages in $[q]$ is said to admit a zero-error code if and only if there exists a non-empty set
$$X_0\subseteq\{x|x\in X^{(m)},m\in[q]\}$$ such that $\forall x,x'\in X_0$, $Y_x\cap Y_{x'}=\varnothing$. The {\em one-shot zero-error capacity} of a classical channel is the number of messages that can be sent without error with one use of the channel, i.e.,
the cardinality of the largest zero-error code it admits. This is given by the independence number $\alpha(G(\mathcal{N}))$ of $G(\mathcal{N})$, namely, the cardinality of the largest set of vertices that share no edge in $G(\mathcal{N})$. Note that $\mathcal{N}$ does {\em not} admit a nontrivial zero-error code (i.e., with $q\geq 2$) if and only if $G(\mathcal{N})$ is a complete graph, i.e., $\alpha(G(\mathcal{N}))=1$. Further, for any encoding $\{X^{(m)}\}_{m=1}^q$ that does admit a zero-error code using $\mathcal{N}$, we necessarily have $q\leq \alpha(G(\mathcal{N}))$. We illustrate the above notions in Fig.~\ref{prevedelchannelhypergraph} with a simple example of a classical channel that was studied in Ref.~\cite{PLM11}.
\begin{figure}[htb!]
	\centering
	\includegraphics[scale=0.2]{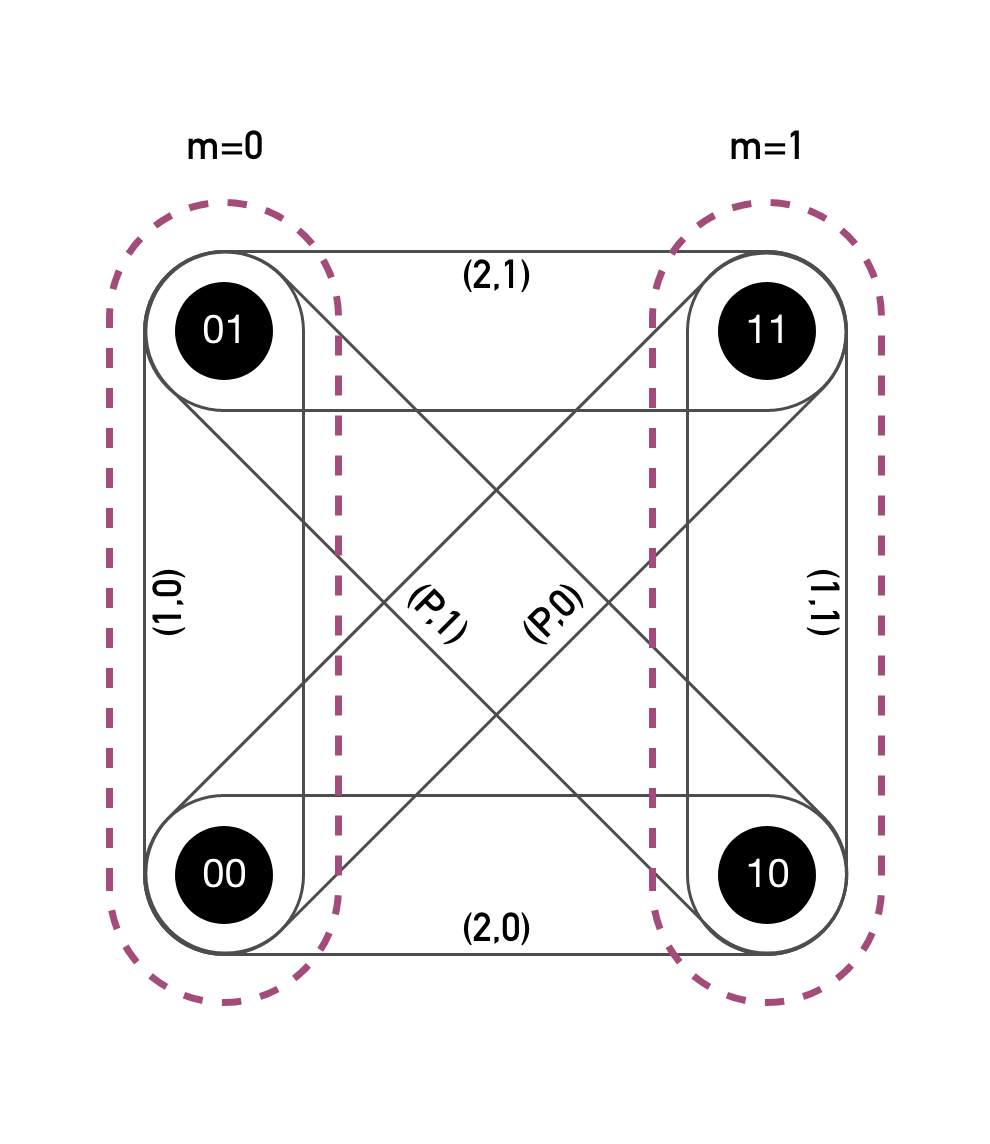}
	\caption{The channel hypergraph of the classical channel studied in Ref.~\cite{PLM11} with the encoding indicated by dashed edges. Here $X=\{00,01,10,11\}$ and $Y=\{(1,0),(1,1),(2,0),(2,1),(P,0),(P,1)\}$. The support of input $01$, for example, is given by $Y_{01}=\{(1,0),(2,1),(P,1)\}\subseteq Y$ and the support of output $(P,0)$, for example, is given by $X_{(P,0)}=\{00,11\}\subseteq X$. The channel probabilities are given by $\mathcal{N}(y|x)=1/3$ for all $y\in Y_x$ and $0$ otherwise for all $x\in X$. The one-bit encoding of $m\in {\rm \bf Msg}=\{0,1\}$ is given by $X^{(0)}=\{00,01\}\subseteq X$ and $X^{(1)}=\{10,11\}\subseteq X$. This encoding $\{X^{(0)}, X^{(1)}\}$ does not admit a zero-error code. In fact, a zero-error code doesn't exist for this channel hypergraph since it does not admit even a pair of inputs that are mutually non-confusable, i.e., do not share an edge.}
	\label{prevedelchannelhypergraph} 
\end{figure}

In this paper we will consider the {\em one-shot success probability} for sending messages in a given encoding $\{X^{(m)}\}_{m=1}^q$ of $q$ messages, where $q>\alpha(G(\mathcal{N}))$. Such an encoding does not admit a zero-error code, although the classical channel $\mathcal{N}$ may admit (smaller) encodings with zero-error codes, i.e., it may be that $\alpha(G(\mathcal{N}))>1$. Of particular interest in this paper is a family of classical channels that we term {\em Kochen-Specker (KS) channels}. A KS channel is defined as any classical channel whose channel hypergraph satisfies the property of {\em KS-uncolourability} \cite{Kunjwal20}\cite{AFL15}, i.e., it is impossible to assign a $\{0,1\}$-valuation to the vertices such that the assignments in each hyperedge add up to 1.\footnote{What this means in terms of channel confusability is that it is impossible to pick a set of vertices with one vertex from each hyperedge such that all the vertices in the set are mutually non-confusable.} A KS-uncolourable hypergraph is said to admit a {\em KS set} if it is possible to associate its vertices to projectors on a Hilbert space and hyperedges to projector-valued measures (PVMs), i.e., the projectors in any hyperedge are mutually orthogonal and sum up to the identity. Any set of such projectors for a KS-uncolourable hypergraph is called a {\em KS set}. We will consider two paradigmatic examples of KS channels, drawing upon Refs.~\cite{CLM10,PLM11}, only one of which \cite{CLM10} admits a KS set.

\subsection{Contextuality}\label{sec2_2}

We will be interested in the twin notions of preparation and measurement noncontextuality following Spekkens \cite{Spekkens05}. In a general operational theory, a preparation procedure consists of a source setting, $\mathbb{S}$, that prepares an ensemble of possible preparations indexed by source outcome $\mathbb{s}\in V_{\mathbb{S}}$, each with probability $p(\mathbb{s}|\mathbb{S})$. We denote the ensemble for source setting $\mathbb{S}$ by $\{(p(\mathbb{s}|\mathbb{S}),[\mathbb{s}|\mathbb{S}])\}_{\mathbb{s}\in V_{\mathbb{S}}}$. Formally, we refer to the (abstract) device implementing this preparation procedure as a multisource. A measurement procedure consists of a measurement setting $\mathbb{M}$ that yields one of possible outcomes indexed by $\mathbb{m}\in V_{\mathbb{M}}$ with probability $p(\mathbb{m}|\mathbb{M},\mathbb{S},\mathbb{s})$ when a system prepared according to $[\mathbb{s}|\mathbb{S}]$ is input to the measurement device. Formally, we refer to the (abstract) device implementing this measurement procedures as a multimeter. Together, the combination of source setting $\mathbb{S}$ and a measurement setting $\mathbb{M}$ yields conditional joint probability distribution given by $p(\mathbb{m},\mathbb{s}|\mathbb{M},\mathbb{S})=p(\mathbb{m}|\mathbb{M},\mathbb{S},\mathbb{s})p(\mathbb{s}|\mathbb{S})$. (See Fig.~\ref{PMsetup}.)

\begin{figure}
	\centering
	\includegraphics[scale=0.37]{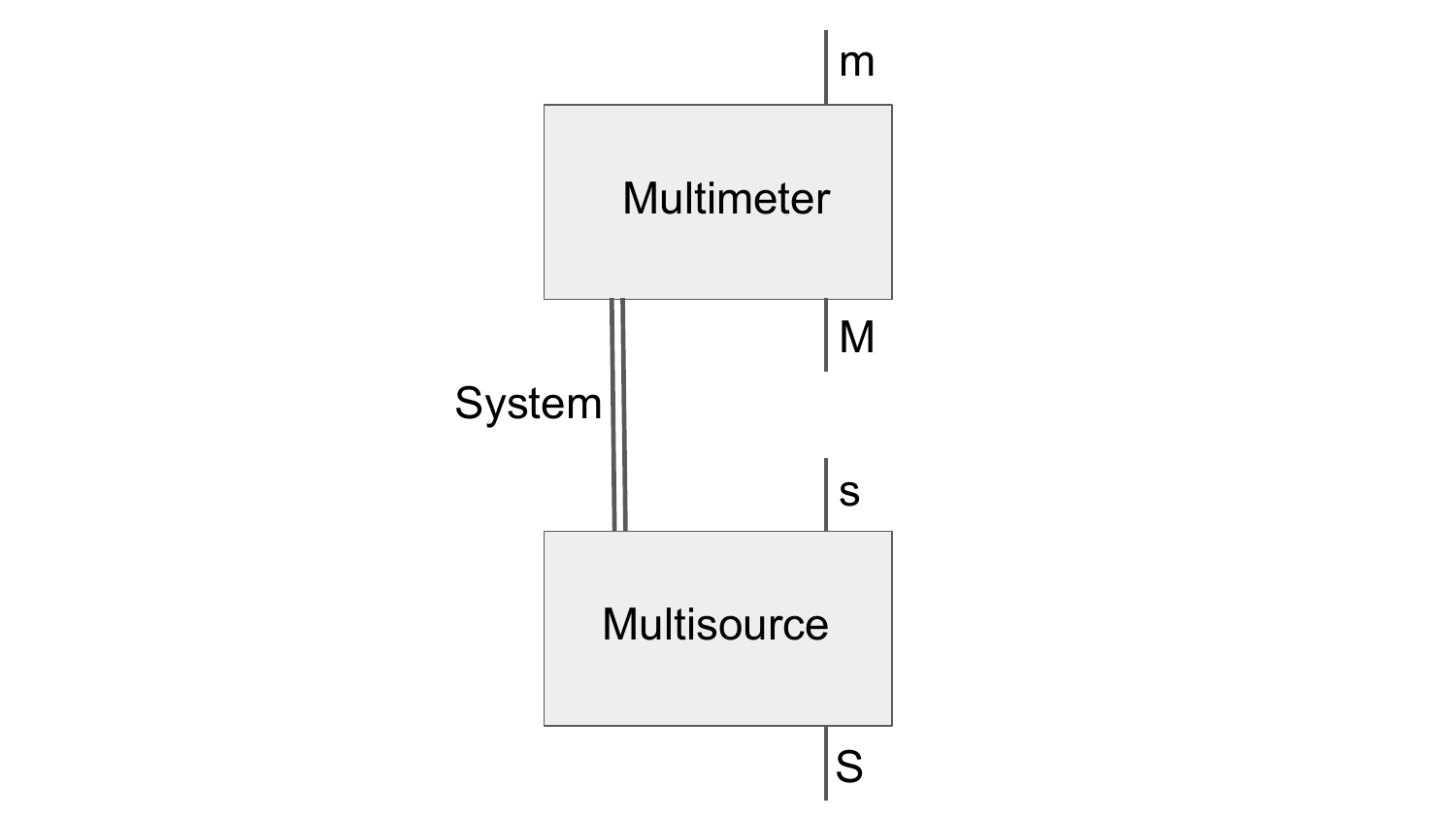}
	\caption{A generic prepare-and-measure setup. Classical inputs and outputs are indicated by single lines and the quantum/postquantum system passing from the multisource to the multimeter is indicated by double lines.}
	\label{PMsetup}
\end{figure}

Two source settings $\mathbb{S}$ and $\mathbb{S'}$ are said to be operationally equivalent if 
\begin{align}\label{prepop}
&\forall [\mathbb{m}|\mathbb{M}]:\nonumber\\
&\sum_{\mathbb{s}}p(\mathbb{m},\mathbb{s}|\mathbb{M},\mathbb{S})=\sum_{\mathbb{s'}}p(\mathbb{m},\mathbb{s'}|\mathbb{M},\mathbb{S'}).
\end{align}
We will denote this operational equivalence by $\mathbb{S}\simeq \mathbb{S'}$. Two measurement events $[\mathbb{m}|\mathbb{M}]$ and $[\mathbb{m'}|\mathbb{M'}]$ are said to be operationally equivalent if 
\begin{align}
\forall [\mathbb{s}|\mathbb{S}]: p(\mathbb{m},\mathbb{s}|\mathbb{M},\mathbb{S})=p(\mathbb{m'},\mathbb{s}|\mathbb{M'},\mathbb{S}),
\end{align}
and we denote this by $[\mathbb{m}|\mathbb{M}]\simeq[\mathbb{m'}|\mathbb{M'}]$.

In keeping with the original definition of generalized noncontextuality \cite{Spekkens05}, and its further development in subsequent work \cite{Kunjwal19}, the notion of operational equivalence---for preparations and for measurements---is evaluated relative to \textit{all} possible measurement and preparation events (respectively) in the operational theory governing the experiment of Fig.~\ref{PMsetup}.

An ontological model of the operational theory consists of ontic states $\lambda\in\Lambda$ that are sampled by a preparation $[\mathbb{s}|\mathbb{S}]$ according to some probability distribution $\mu(\lambda|\mathbb{S},\mathbb{s})$, so that $\mu(\lambda,\mathbb{s}|\mathbb{S})=\mu(\lambda|\mathbb{S},\mathbb{s})p(\mathbb{s}|\mathbb{S})$. Any measurement device responds to the input of an ontic state $\lambda$ according to some probability, $\xi(\mathbb{m}|\mathbb{M},\lambda)$, called a response function. The ontological model reproduces the operational statistics as follows:
\begin{align}
p(\mathbb{m},\mathbb{s}|\mathbb{M},\mathbb{S})=\sum_{\lambda}\xi(\mathbb{m}|\mathbb{M},\lambda)\mu(\lambda,\mathbb{s}|\mathbb{S}).
\end{align} 

The assumption of preparation noncontextuality entails the following implication:
\begin{align}
\mathbb{S}\simeq \mathbb{S'}&
\Rightarrow
\sum_{\mathbb{s}}\mu(\lambda,\mathbb{s}|\mathbb{S})=\sum_{\mathbb{s'}}\mu(\lambda,\mathbb{s'}|\mathbb{S'}), \forall \lambda\in \Lambda,\nonumber\\
&\textrm{i.e., }\mu(\lambda|\mathbb{S})=\mu(\lambda|\mathbb{S'}),\forall \lambda\in \Lambda.
\end{align}

The assumption of measurement noncontextuality entails the following implication:
\begin{align}
[\mathbb{m}|\mathbb{M}]\simeq [\mathbb{m'}|\mathbb{M'}]\Rightarrow \xi(\mathbb{m}|\mathbb{M},\lambda)=\xi(\mathbb{m'}|\mathbb{M'},\lambda), \forall \lambda\in\Lambda.
\end{align}

A failure of the joint assumption of preparation and measurement noncontextuality is then said to be a demonstration of contextuality.\footnote{The assumption of noncontextuality is an instance of the Leibnizian idea of the ontological identity of operational indiscernibles \cite{Spekkens19}: operationally equivalent experimental procedures admit ontologically equivalent representations under this assumption.}
\section{General protocol for one-shot classical communication assisted by a common-cause resource}\label{sec3}
\subsection{One-shot success probability}\label{sec3_1}
We want to consider the situation where Alice and Bob have access to a shared common-cause resource such as quantum entanglement (but possibly also more general nonclassical common-cause resources \cite{WSS20,Barrett07}) which they can use to enhance the one-shot success probability of sending messages through $\mathcal{N}$. This situation of entanglement-assisted one-shot classical communication has been studied previously \cite{CLM10, PLM11, HMS13}. We will below take the nonclassical common-cause resource to be quantum entanglement for ease of presentation but the ideas extend to post-quantum theories -- corresponding, in general, to some convex subset of the set of no-signalling correlations -- in a straightforward way.

The task at hand is the following: Alice and Bob share a classical channel $\mathcal{N}$ together with some bipartite quantum system in an entangled state $\rho_{\rm AB}$. Alice wants to send messages from the set ${\rm \bf Msg}$ according to a probability distribution $\{p(m)\}_{m\in \rm \bf \bf Msg}$. To do this for some encoding $\{X^{(m)}\}_{m\in \rm \bf Msg}$ defined on $\mathcal{N}$, Alice implements a POVM $\mathbb{M}^{\rm (A)}_m\equiv\{E^{(m)}_x\}_{x\in X}$ on her part, $\rho_A$, of the shared state $\rho_{AB}$ given by $\rho_A=\Tr_B\rho_{AB}$. She obtains outcome $x\in X$ with probability $p(x|m)=\Tr(E^{(m)}_x\rho_A)$. She inputs $x$ into the classical channel which yields output $y\in Y$ with probability $\mathcal{N}(y|x)$. Using the output $y$ and his part of $\rho_{\rm AB}$, Bob needs to figure out a strategy that will let him infer Alice's choice of measurement (hence the message $m$) with the maximum success probability, i.e., for every set of $\{m,x,y\}$ we want to maximize the probability $p(m'=m|y,m,x)$ that Bob's guess for the message, denoted $m'$, agrees with the message Alice sent.

On receiving the channel output $y$, Bob implements some measurements, say $\{\mathbb{M}_{v}\equiv\{E^{(v)}_{z}\}_{z\in O}\}_{v\in V}$, on his part of the shared quantum system according to some probability distribution, say $\{p(v|y)\}_{v\in V}$, that depends on $y$. The measurement outcome $z\in O$ occurs with probability $p(z|v,m,x)=\Tr (E^{(v)}_{z}\rho^{(m)}_x)$ for the choice of POVM $\mathbb{M}_v$. Here $\rho^{(m)}_x$ is the state on Bob's side that Alice ``steers" to when she obtains outcome $x$ for measurement $m$ with probability $p(x|m)$. Overall, Bob implements the effective POVM $\mathbb{M}'_y\equiv \sum_vp(v|y)\mathbb{M}_v$ given by the set of POVM elements $\{\sum_vp(v|y)E^{(v)}_z\}_{z\in O}$, where outcome $z$ occurs with probability $p(z|y,m,x)\equiv \sum_v p(v|y)p(z|v,m,x)$. Bob's guess for the message, $m'$, will then be a function of $z,y$, i.e., $m'=g(z,y)$, where $g$ is a function from $O\times Y$ to ${\rm \bf Msg}$. We have a successful decoding when $m'=m$. This effectively defines the overall measurement $\mathbb{M}^{\rm (B)}_y$ with outcomes $m'\in {\rm \bf Msg}$, obtained by classical post-processing of the outcomes $z\in O$ of $\mathbb{M}'_y$, according to $p(m'|y,m,x)\equiv \sum_z\delta_{m',g(z,y)}p(z|y,m,x)$. The probability of a correct guess, $m'=m$, is then $p(m'=m|y,m,x)=\sum_z\delta_{m,g(z,y)}p(z|y,m,x)$.

The overall success probability is thus given by
\begin{align}
S=&\sum_mp(m)\sum_xp(x|m)\sum_y\mathcal{N}(y|x)\nonumber\\
&\sum_{m'}p(m'|y,m,x)\delta_{m,m'}.
\end{align}

In situations where Bob can do this perfectly (the zero-error regime), i.e., $S=1$, we obtain an exact simulation of a noiseless classical channel ${\rm \bf Id}$ with the same input and output alphabet ${\rm \bf Msg}$ with channel probabilities ${\rm\bf Id}(m'|m)=\delta_{m,m'}$ by using the noisy classical channel $\mathcal{N}$ and (potentially) some shared common-cause resource. A schematic of the protocol is provided in Fig.~\ref{schematic}.

\subsection{KS channels that admit KS sets: Cubitt {\em et al.} strategy}\label{sec3_2}
For some KS channels admitting KS sets (e.g., Fig.~\ref{kschannel}), the Cubitt {\em et al.}~strategy \cite{CLM10} corresponds to choosing ${\rm \bf Msg}\equiv\{m\}_{m=1}^q$, $V\equiv Y$, $p(v|y)\equiv \delta_{v,y}$, $O\equiv X$ for all $v\in V$, (hence) $z\equiv x'\in X$, and $g(z,y)=g(x')=m'$, where $m'$ is the unique message such that $x'\in X^{(m')}$. Thus, $p(m'|y,m,x)=\sum_{x'} \delta_{m',g(x')}p(x'|y,m,x)$. With these choices, we have

\begin{align}
S=&\sum_mp(m)\sum_xp(x|m)\sum_y\mathcal{N}(y|x)\nonumber\\
&\sum_{m'}p(m'|y,m,x)\delta_{m,m'}\nonumber\\
=&\sum_mp(m)\sum_xp(x|m)\sum_y\mathcal{N}(y|x)\nonumber\\
&\sum_{x'\in X^{(m)}}p(x'|y,m,x).
\end{align}

As shown in Ref.~\cite{CLM10}, such channels admit an enhancement of their one-shot zero-error classical capacity in the presence of shared entanglement, i.e., their one-shot success probability $S=1$ for some $q>\alpha(G(\mathcal{N}))$. We will discuss these channels in more detail in Section \ref{sec7}.

\subsection{KS channels that do not admit KS sets}\label{sec3_3}
In general, a KS channel may not admit any KS set and in that case the strategy of Cubitt {\em et al.} \cite{CLM10} does not apply.  An example of a KS channel that does not admit a KS set was studied for its one-shot success probability by Prevedel {\em et al.} \cite{PLM11}. This example fits within the general protocol described in Section \ref{sec3_1} and will be discussed in Section \ref{prevedelprotocol}.

\section{Quantum advantage in one-shot classical communication}\label{sec4}

\subsection{Classical one-shot success probability}
Classically, the only information that Bob has about Alice's measurement and its outcome, i.e., $m$ and $x$, is mediated by the output of the channel, $y$. We therefore have $p(m'|y,m,x)=p(m'|y)$, i.e., $m'$ is conditionally independent of $m$ and $x$, given $y$. Shared randomness does not help because it only amounts to a convex mixture of deterministic classical strategies (indexed by, say, $c\in\mathcal{C}$) according to some probability distribution $\{p(c)\}_{c\in\mathcal{C}}$ and no such convex mixture can do better than the best deterministic classical strategy. The classical one-shot success probability is therefore given by
\begin{align}
&S_{\rm Cl}\nonumber\\
=&\sum_{c\in\mathcal{C}}p(c)\Bigg(\sum_mp(m)\sum_xp(x|m,c)\sum_y\mathcal{N}(y|x)\nonumber\\
&\sum_{m'}p(m'|y,c)\delta_{m,m'}\Bigg),
\end{align}
with the tight upper bound 
\begin{align}
&S_{\rm Cl}^{\max}\equiv\max_{c\in\mathcal{C}}\Bigg(\sum_mp(m)\sum_xp(x|m,c)\sum_y\mathcal{N}(y|x)\nonumber\\
&\sum_{m'}p(m'|y,c)\delta_{m,m'}\Bigg),
\end{align}
so that 
\begin{equation}
	S_{\rm Cl}\leq S_{\rm Cl}^{\max}.
\end{equation}

\subsection{Preparation contextuality drives the quantum advantage}\label{prepctxnecc}

In the protocol we have described, the following operational equivalences hold: $\sum_x p(x|m_1)\rho^{(m_1)}_x=\sum_x p(x|m_2)\rho^{(m_2)}_x\equiv \rho_B$ for all pairs of distinct messages $m_1,m_2\in {\rm \bf Msg}$.\footnote{Although we are using quantum notation here, operationally, this prepare-and-measure setup on Bob's system (see Fig.~\ref{PMsetup}) can be viewed as a multisource with settings $\mathbb{S}=m\in {\rm \bf Msg}$ and outcomes $\mathbb{s}=x\in X$ that occur with probability $p(\mathbb{s}|\mathbb{S})=p(x|m)$. The operational equivalences can then be expressed as in Eq.~\ref{prepop}, i.e., $\forall [\mathbb{m}|\mathbb{M}]$, $\sum_{\mathbb{s}_1}p(\mathbb{m},\mathbb{s}_1|\mathbb{M},\mathbb{S}_1)=\sum_{\mathbb{s}_2}p(\mathbb{m},\mathbb{s}_2|\mathbb{M},\mathbb{S}_2)$, once we compute the probability $p(\mathbb{m}|\mathbb{M},\mathbb{S},\mathbb{s})=\Tr(\rho^{(\mathbb{S})}_{\mathbb{s}} E^{(\mathbb{M})}_{\mathbb{m}})$ for all effects $E^{(\mathbb{M})}_{\mathbb{m}}$ representing measurement events $[\mathbb{m}|\mathbb{M}]$. That is, $\sum_x p(x|m_1)\rho^{(m_1)}_x=\sum_x p(x|m_2)\rho^{(m_2)}_x$ translates to $\sum_x p(x|m_1)\Tr(\rho^{(m_1)}_x E^{(\mathbb{M})}_{\mathbb{m}})=\sum_x p(x|m_2)\Tr(\rho^{(m_2)}_x E^{(\mathbb{M})}_{\mathbb{m}})$ for all $[\mathbb{m}|\mathbb{M}]$. The multimeter has settings $\mathbb{M}$ and outcomes $\mathbb{m}$ that range over the set of all measurement events in the operational theory. Note, however, that as long as we assume the shared common-cause resource is non-signalling (as is the case with entangled states in quantum theory), we do not need to explicitly verify these operational equivalences by varying over all measurement events: these equivalences are implied by the non-signalling nature of the shared common-cause resource.} This follows from the fact that the common-cause correlations shared between Alice and Bob must be nonsignalling: Alice's choice of POVM $\mathbb{M}_{m}$ encoding the message $m$ ($m\in{\rm \bf Msg}$) steers Bob's system to the ensemble of states $\{(p(x|m),\rho^{(m)}_x)\}_{x\in X^{(m)}}$; however, on coarse-graining, the reduced state on Bob's side, $\rho_{\rm B}=\Tr_{\rm A}\rho_{\rm AB}$, is the same for all choices of $m$, and thus the common-cause correlations cannot be used by Bob to infer $m$.

Preparation noncontextuality then entails that 
$$\sum_{\mathbb{s}} p(x|m_1)p(\lambda|m_1,x)=\sum_x p(x|m_2)p(\lambda|m_2,x)$$

$$\sum_x p(x|m_1)p(\lambda|m_1,x)=\sum_x p(x|m_2)p(\lambda|m_2,x)$$ for all $\lambda\in \Lambda$, for all $m_1,m_2\in{\rm \bf Msg}$. That is, $p(\lambda|m_1)=p(\lambda|m_2)\equiv p(\lambda)$ for all $\lambda\in \Lambda$, where $\Lambda$ is the ontic state space of Bob's system.

Given the prior distribution $\{p(m)\}_
{m\in{\rm\bf Msg}}$ and the channel probabilities $\{\mathcal{N}(y|x)\}_{y,x}$, we obtain, under the assumption of preparation noncontextuality for Bob's system, the following expression for the one-shot success probability and the preparation noncontextual upper bound on it:
\begin{align}
&S_{\rm PNC}\nonumber\\
=&\sum_{\lambda}p(\lambda)\Bigg(\sum_mp(m)\sum_xp(x|m,\lambda)\sum_y \mathcal{N}(y|x)\nonumber\\
&\sum_{m'}p(m'|y,\lambda)\delta_{m,m'}\Bigg)\nonumber\\
\leq&\max_{\lambda}\Bigg(\sum_mp(m)\sum_xp(x|m,\lambda)\sum_y\mathcal{N}(y|x)\nonumber\\
&\sum_{m'}p(m'|y,\lambda)\delta_{m,m'}\Bigg)\equiv S_{\rm PNC}^{\rm max}.
\end{align}
To see how this comes about, note that the one-shot success probability,
\begin{align}
S=&\sum_mp(m)\sum_xp(x|m)\sum_y\mathcal{N}(y|x)\nonumber\\
&\sum_{m'}p(m'|y,m,x)\delta_{m,m'},\nonumber
\end{align}
when expressed in terms of an ontological model for Bob's system, requires that 
$p(m'|y,m,x)=\sum_{\lambda}p(m'|y,\lambda)p(\lambda|m,x)$. We can then write a joint probability distribution $p(x,\lambda|m)=p(\lambda|m,x)p(x|m)$, which can be rewritten as $p(x,\lambda|m)=p(x|m,\lambda)p(\lambda|m)$. Recalling that preparation noncontextuality requires $p(\lambda|m)=p(\lambda)$ for all $m$, we obtain the expression for $S_{\rm PNC}$. Hence, we have: $S_{\rm PNC}^{\max}=S_{\rm Cl}^{\max}$.

We now argue that the upper bound $S^{\max}_{\rm PNC}$ can be saturated by a preparation noncontextual ontological model. Such a model (achieving $S_{\rm PNC}=S_{\rm PNC}^{\max}$) must necessarily have $p(\lambda|m)=p(\lambda)=\delta_{\lambda,\lambda_{\max}}$ for all $m\in{\rm \bf Msg}$, where $\lambda_{\max}$ is the\footnote{For our purposes, we can take $\lambda_{\max}$ to be unique without any loss of generality.} ontic state of Bob's system ($\lambda_{\max}\in\Lambda$) that satisfies 
\begin{align}
S_{\rm PNC}^{\max}=&\sum_mp(m)\Bigg(\sum_xp(x|m,\lambda_{\max})\sum_y\mathcal{N}(y|x)\nonumber\\
&\sum_{m'}p(m'|y,\lambda_{\max})\delta_{m,m'}\Bigg)
\end{align}
Thus, $p(\lambda|m)=\sum_x p(x|m)p(\lambda|m,x)=\delta_{\lambda,\lambda_{\max}}$ implies that $p(\lambda|m,x)=\delta_{\lambda,\lambda_{\max}}$ for all $m,x$. 
We must then have $p(m'|y,m,x)=\sum_{\lambda}p(m'|y,\lambda)p(\lambda|m,x)=\sum_{\lambda}p(m'|y,\lambda)\delta_{\lambda,\lambda_{\max}}=p(m'|y,\lambda_{\max})$, i.e., the statistics of $y$ does not change in response to variations in $m$ and $x$ but is directly determined by the ontic state that is deterministically sampled by {\em every} preparation procedure. In fact, the response functions cannot even deviate from the best deterministic classical strategy. This preparation noncontextual ontological model, therefore, trivially reproduces the operational equivalence required on Bob's system, i.e., the operational equivalence between all the coarse-grainings of preparation ensembles induced by Alice's measurements. It also achieves $S=S^{\max}_{\rm PNC}$ by fixing the response functions on Bob's side to mimic the best deterministic classical strategy.\footnote{Hence, the ontological model can only simulate an operational theory that has just one equivalence class of preparations and, furthermore, associates outcomes to its measurements deterministically. As such, in the presence of other empirical facts that an operational theory might present (such as the simple fact that Bob's system can be prepared in operationally inequivalent ways), this preparation noncontextual model will fail to reproduce predictions of the theory that go beyond the required operational equivalence between preparation procedures. Generically, therefore, any non-trivial (at least in the sense of admitting operationally inequivalent preparation procedures) preparation noncontextual ontological model will only achieve $S<S^{\max}_{\rm PNC}$.}

Hence, we have that $S\leq S^{\max}_{\rm Cl}$ {\em is} a preparation noncontextuality inequality and any quantum advantage in this communication task witnesses preparation contextuality. 

\subsection{The shared state must violate a Bell inequality under the local measurements used in the communication protocol}\label{bellnlnecc}
In a quantum implementation of the communication protocol, shared entanglement between Alice and Bob is crucial for there to be an advantage over the classical one-shot success probability. However, entanglement alone is not enough: the entanglement must be such that it enables a Bell inequality violation relative to the local measurements that Alice and Bob implement. To see this, consider a nonlocal game that uses the same resources---shared entanglement and local measurements---as the communication protocol but under spacelike separation (hence no classical channel): Alice and Bob implement their local measurements labelled by $m$ and $y$, respectively, and obtain their respective outcomes $x$ and $m'$ with a joint probability $p(x,m'|m,y)$ and the joint statistics thus collected admits a locally causal model, i.e.,
\begin{align}
p(x,m'|m,y)=\sum_{\omega\in \Omega}p(x|m,\omega)p(m'|y,\omega)p(\omega),
\end{align}
where $\omega$ denotes the shared ontic state sampled from the ontic state space $\Omega$ of the bipartite system Alice and Bob share. (Note that this is, in general, different from the ontic state space $\Lambda$ for Bob's system alone.) In such a case, it is straightforward to see that the achievable success probability is no better than the best deterministic classical strategy. Firstly,
\begin{align}
S=&\sum_mp(m)\sum_xp(x|m)\sum_y\mathcal{N}(y|x)\nonumber\\
&\sum_{m'}p(m'|y,m,x)\delta_{m,m'}\nonumber\\
=&\sum_mp(m)\sum_x\sum_y\mathcal{N}(y|x)\nonumber\\
&\sum_{m'}p(x|m)p(m'|y,m,x)\delta_{m,m'}\nonumber\\
=&\sum_mp(m)\sum_x\sum_y\mathcal{N}(y|x)\sum_{m'}p(x,m'|m,y)\delta_{m,m'},
\end{align}
which can be now be expressed as 
\begin{align}
S=&\sum_mp(m)\sum_x\sum_y\mathcal{N}(y|x)\sum_{m'}\Bigg(\sum_{\omega\in \Omega}p(x|m,\omega)\nonumber\\
&p(m'|y,\omega)p(\omega)\Bigg)\delta_{m,m'}\nonumber\\
=&\sum_{\omega\in \Omega}p(\omega)\Bigg(\sum_mp(m)\sum_xp(x|m,\omega)\sum_y\mathcal{N}(y|x)\nonumber\\
&\sum_{m'}p(m'|y,\omega)\delta_{m,m'}\Bigg)\nonumber\\
\leq&S^{\rm Cl}_{\max}.\label{localbound}
\end{align}
Hence, Bell nonlocality of the shared entangled state relative to the measurements carried out by Alice and Bob is a necessary condition for a quantum advantage in this communication task.\footnote{Note, however, that -- contrary to the counterfactual Bell scenario we just considered -- the measurement choice $y$ in the communication protocol of interest is not free and is determined probabilistically by Alice's measurement outcome $x$. That is, the protocol requires a wiring of Alice's system with Bob's using the classical channel, something distinct from a Bell scenario.}

\section{Preparation contextuality vis-\`a-vis Bell nonlocality: the connection with nonlocal games}\label{sec5}

It is known that any bipartite proof of Bell nonlocality can be turned into a proof of preparation contextuality on each wing of the Bell experiment, i.e., Bell nonlocality implies preparation contextuality on both wings of the Bell experiment. It is easiest to see this in the contrapositive: that is, the existence of a preparation noncontextual ontological model on any wing of the Bell experiment implies the existence of a locally causal model for the Bell experiment. We provide an explicit argument in Appendix \ref{pncimplieslc}.\footnote{On the other hand, in the simplest possible scenario capable of exhibiting preparation contextuality (with two tomographically complete binary measurements and four preparations), it has been shown that the existence of a preparation noncontextual ontological model is {\em equivalent} to the existence of a locally causal model in any bipartite extension of the one-party scenario to a CHSH scenario \cite{Pusey18}. Ref.~\cite{Pusey18} also noted that Barrett was the first to show the implication from preparation noncontextuality to local causality in any bipartite extension of a given one-party scenario. This has also been observed in Ref.~\cite{TU20}}

For the task of enhancing the one-shot success probability of a classical channel, preparation contextuality and Bell nonlocality are even more intimately related than the general situation above. As we just showed in Section \ref{bellnlnecc}, $S>S_{\rm Cl}^{\max}$ implies Bell nonlocality of the joint statistics $\{p(x,m'|m,y)\}_{x,m',m,y}$. This allows us to state the following proposition:

\begin{proposition}\label{nonlocalgamexistence}
For every classical channel $\mathcal{N}$ that admits an enhancement of the one-shot success probability driven by preparation contextuality, i.e., $S>S_{\rm Cl}^{\max}$, there exists a nonlocal game which can be  won with a better-than-classical success probability by the same entangled state and local measurements which enable an advantage in the communication task. By construction, we also have the converse: an advantage in this nonlocal game would imply an advantage in the communication task.
\end{proposition}

Indeed, if this were not the case (i.e., no such nonlocal game existed) then the enhancement of the one-shot success probability couldn't have been exhibited because the shared correlations between Alice and Bob would then be Bell-local. Hence, the problem of one-shot classical communication assisted by non-signalling correlations characterizes a family of Bell scenarios where a proof of preparation contextuality on one wing implies a proof of Bell nonlocality between the two wings. This provides further insight into the conditions under which preparation contextuality for a single system can be said to imply Bell nonlocality for its appropriate bipartite extensions, in line with previous work where a certain type of preparation contextuality was shown to imply Bell nonlocality relative to a bipartite extension \cite{LM13}. \footnote{In the case of bipartite pure entangled states of Schmidt rank greater than two (such as the two-ququart maximally entangled state used in Ref.~\cite{CLM10}), it has been shown that preparation contextuality of the reduced state on Bob's system, in conjunction with Alice's ability to remotely steer Bob's system to arbitrary preparation ensembles using entanglement \cite{HJW93}, implies that  the entangled state can exhibit Bell nonlocality \cite{BBC14}. So, at least in the case of such pure entangled states, the implication from preparation contextuality to Bell nonlocality that we consider in this paper also follows from Ref.~\cite{BBC14}.} 

The explicit construction of a nonlocal game instantiating Proposition \ref{nonlocalgamexistence}, however, would depend on the properties of the channel $\mathcal{N}$. We know that at least in the case of the Cubitt {\em et al.}~protocol under ideal conditions, these nonlocal games correspond to pseudo-telepathy (PT) games inspired by the KS theorem \cite{CLM10, HR83, BBT05}. In the case of the Prevedel {\em et al.}~example \cite{PLM11}, the associated nonlocal game is essentially the well-known CHSH game \cite{CHSH, BCP14}. It is an open question whether there exists a generic construction of a nonlocal game, instantiating Proposition \ref{nonlocalgamexistence}, that always works starting from any channel $\mathcal{N}$.

We {\em can}, however, provide a fairly general construction of nonlocal games starting from a family of classical channels, thus instantiating Proposition \ref{nonlocalgamexistence}. This general construction, in particular, reproduces as special cases the examples studied in Refs.~\cite{CLM10,PLM11}. It is inspired by the pseudotelepathy (PT) game discussed in Ref.~\cite{CLM10}, allowing, however, the case where the quantum strategy is imperfect and where neither Alice nor Bob might have access to a KS set. We define this mapping from the communication task to a nonlocal game below.

The family of classical channels for which our construction works satisfies two properties for any channel $\mathcal{N}$ in the family: first, its channel hypergraph $H(\mathcal{N})$ is $k$-regular (i.e., every vertex appears in $k$ hyperedges) for some positive integer $k$, and second, the channel probabilities are {\em entirely} fixed by the combinatorial structure of the channel, i.e., $\mathcal{N}(y|x)=\frac{1}{|Y_x|}\delta(y\in Y_x)$, where $\delta(a\in A)$ defines an indicator function for membership in set $A$, taking value $1$ if $a\in A$ and $0$ otherwise. Channels satisfying the first property will be called {\em $k$-regular} channels and those satisfying the second property will be called {\em output-uniform} channels. Hence, the classical channels we consider below will be $k$-regular and output-uniform classical channels. Further, we will assume that Alice's choice of the message to send in a particular run is uniformly random, i.e., $p(m)=\frac{1}{|{\rm\bf Msg}|}$ for all $m\in {\rm\bf Msg}$. All this amounts to the following expression for the one-shot success probability:
\begin{align}\label{oneshotexp}
S=\frac{1}{|{\rm \bf Msg}|}\frac{1}{k}\sum_{m,m',x,y\in Y_x}\delta_{m,m'}p(x,m'|m,y).
\end{align}

The corresponding nonlocal game is specified by the following: Alice receives questions $m\in {\rm \bf Msg}$ and replies with answers $x\in X$; Bob receives questions $y\in Y$ and replies with answers $m'\in  {\rm \bf Msg}$; the conditional joint probability distributions of interest, therefore, are given by $\{p(x,m'|m,y)\}_{x,m',m,y}$; the Referee sends them questions $m,y$ according to the probability distribution $p(m,y)=p(m)p(y)=\frac{1}{|{\rm \bf Msg}|}\frac{1}{|Y|}$; in order to win the game, Alice and Bob must produce outputs $x,m'$ (respectively) such that the condition $V(x,m',m,y)=1$ is satisfied, where  
\begin{align}\label{gamerules}
V(x,m',m,y)\equiv\begin{cases}
1, \textrm{ if }y\notin Y_x \textrm{ (i.e., }\mathcal{N}(y|x)=0\textrm{)},\\
\delta_{m,m'},\textrm{ if }y\in Y_x \textrm{ (i.e., }\mathcal{N}(y|x)>0\textrm{)}.
\end{cases}
\end{align}
The probability of winning the game is then given by 
\begin{align}
S_{\rm Bell}\equiv&\frac{1}{|{\rm \bf Msg}|}\frac{1}{|Y|}\sum_{m,m',x,y}p(x,m'|m,y)V(x,m',m,y)\\
=&\frac{1}{|{\rm \bf Msg}|}\frac{1}{|Y|}\sum_{m,m',x,y\in Y_x}\delta_{m,m'}p(x,m'|m,y)\nonumber\\
+&\frac{1}{|{\rm \bf Msg}|}\frac{1}{|Y|}\sum_{m,m',x,y\notin Y_x}p(x,m'|m,y).\label{bellexp}
\end{align}

Note that this mapping relies only on combinatorial properties of $\mathcal{N}$, namely, its channel hypergraph $H(\mathcal{N})$, and is a straightforward generalization of the connection between the one-shot classical communication protocol and pseudo-telepathy games note in Proposition 3 of Ref.~\cite{CLM10}. The connection of the protocol of Ref.~\cite{PLM11} with the CHSH game, for example, falls under this generalization.

We are now ready to prove the following theorem:

\begin{theorem}\label{isomorphism}
	Consider a $k$-regular output-uniform classical channel $\mathcal{N}$, i.e., $\mathcal{N}(y|x)=\frac{1}{k}\delta(y\in Y_x)$. Then, for any bipartite entangled state $\rho_{\rm AB}$ shared between Alice and Bob, Alice's local measurements $\{\mathbb{M}^{\rm (A)}_m\equiv\{E^{(m)}_x\}_{x\in X^{(m)}}\}_{m\in{\rm\bf Msg}}$, and Bob's local measurements $\{\mathbb{M}^{\rm (B)}_y\equiv\{E^{(y)}_{m'}\}_{m'\in {\rm \bf Msg}}\}_{y\in Y}$ --- all denoted by the triplet $(\rho_{\rm AB},\{\mathbb{M}^{\rm (A)}_m\}_{m\in{\rm\bf Msg}}, \{\mathbb{M}^{\rm (B)}_y\}_{y\in Y})$ ---
	the following are equivalent:
	
	\begin{enumerate}
		
		\item The triplet $(\rho_{\rm AB},\{\mathbb{M}^{\rm (A)}_m\}_{m\in{\rm\bf Msg}}, \{\mathbb{M}^{\rm (B)}_y\}_{y\in Y})$ provides a quantum advantage in the task of one-shot classical communication over $\mathcal{N}$, i.e., $S>S^{\max}_{\rm Cl}$,
		
		\item The triplet $(\rho_{\rm AB},\{\mathbb{M}^{\rm (A)}_m\}_{m\in{\rm\bf Msg}}, \{\mathbb{M}^{\rm (B)}_y\}_{y\in Y})$ provides a quantum advantage in the corresponding nonlocal game (cf.~Eq.~\eqref{gamerules}), i.e., $S_{\rm Bell}>S_{\rm local}^{\max}$, where $S_{\rm local}^{\max}$ is the Bell-local bound given by
		\begin{align}
		S_{\rm local}^{\max}\equiv\max_{p(x,m'|m,y)\in\mathcal{L}}S_{\rm Bell},
		\end{align}
		$\mathcal{L}$ being the set of Bell-local probability distributions. We use ``$p(x,m'|m,y)\in\mathcal{L}$" as shorthand for membership of the full probability vector $(p(x,m'|m,y))_{x,m',m,y}$ in the set of Bell-local probability vectors \cite{BCP14}.
	\end{enumerate}	
Further, we have that $S=1\Leftrightarrow S_{\rm Bell}=1$.\footnote{Note that this recovers the special case of pseudotelepathy games considered in Proposition 3 of Ref.~\cite{CLM10}.}
\end{theorem}
\begin{proof} See Appendix \ref{proofisomorphism}.
\end{proof}

Theorem \ref{isomorphism} provides us a way to characterize a family of classical channels $\mathcal{N}$ for which the one-shot success probability can be enhanced by nonsignalling correlations: namely, all $k$-regular and output-uniform classical channels for which there is a gap between the classical and the nonsignalling value of the nonlocal game defined by them following the recipe we have just outlined, \textit{cf.}~Eq.~\eqref{gamerules}.

It is worth emphasizing here the physical distinction between the two tasks -- the one-shot communication task and the corresponding nonlocal game -- we have considered in this section. In the communication task, Alice and Bob must necessarily be timelike separated, but in the nonlocal game, they must necessarily be spacelike separated. In the absence of spacelike separation in the communication task, it is inaccurate to state that Bell nonlocality drives the quantum advantage in the task: to be sure, the states and measurements that drive the quantum advantage in the communication task can also drive the quantum advantage in the nonlocal game, but the two tasks correspond to fundamentally different physical situations. Contextuality, for this reason, is the more natural notion of nonclassicality to appeal to as the driver of quantum advantage in the communication task.

\section{One-shot success probability of communicating a single bit}\label{sec6}

The problem of communicating a single bit through a noisy classical channel has been previously studied in Refs.~\cite{PLM11,HMS13}. Note that the classical value of the one-shot success probability of communicating a single bit (i.e., one out of two messages) in this problem is strictly less than $1$ if and only if the confusability graph of the channel $\mathcal{N}$ is a complete graph, i.e., $\alpha(G(\mathcal{N}))=1$. Hence, it is only for such channels that the possibility of enhancing their one-shot success probability of sending a single bit using shared entanglement exists: for all other channels (with $\alpha(G(\mathcal{N}))\geq 2$), a single bit can always be sent with zero error classically. In the rest of this section, therefore, we will only consider output-uniform channels with a confusability graph that is complete. These channels are a special case of Kochen-Specker (KS) channels, namely, those where it is impossible to pick even a {\em pair} of non-confusable vertices from distinct hyperedges.\footnote{Recall that a KS channel is defined by a channel hypergraph where it is impossible to pick a {\em set} of vertices, one vertex from {\em each} hyperedge, such that all the vertices in this set are mutually non-confusable.} Such KS channels obviously do not admit any KS sets since any set of projectors associated with their vertices will necessarily have to be pairwise commuting, i.e., there would be no incompatibility between the projectors. Our general protocol applies to such channels and here we will consider one particular example, the one studied in Ref.~\cite{PLM11}, as a paradigmatic case and show how it fits within our framework.

\subsection*{The Prevedel {\em et al.}~protocol}\label{prevedelprotocol}

The classical channel considered by Prevedel {\em et al.}~\cite{PLM11} is specified as follows: the input alphabet is the set of two-bit strings, $X=\{00,01,10,11\}$, and the output alphabet is a set of trit-bit strings $Y=\{(1,0),(2,0),(P,0),(1,1),(2,1),(P,1)\}$; the channel hypergraph therefore consists of $4$ vertices (labelled by $x=b_1b_2\in X$) with all possible two-vertex hyperedges, i.e., $6$ hyperedges labelled by $y=(t,b)\in Y$; each hyperedge $y=(t,b)$ uniquely identifies a pair of inputs $x\in X_y$ sharing one of three properties: the value of the first bit (i.e., $y=(1,b_1)$), the value of the second bit (i.e., $y=(2,b_2)$), or the parity of the two bits (i.e., $y=(P,b_1\oplus b_2)$); so, for example, the output $y=(1,0)$ identifies the pair of inputs $X_y=\{00,01\}$ in its support; the channel is, therefore, $3$-regular and output-uniform with channel probabilities $\mathcal{N}(y|x)=\frac{1}{3}\delta(y\in Y_x)$. The channel hypergraph is illustrated in Fig.~\ref{prevedelchannelhypergraph}.

The general protocol of Section \ref{sec3} takes the following form: 
\begin{align}
&{\rm \bf Msg}=\{m\}_{m=0}^1, \textrm{ with the encoding } \{X^{(m)}\}_m \textrm{ given by }\nonumber\\
&X_{m=0}\equiv\{00, 01\}\textrm{ and } X_{m=1}\equiv\{10, 11\}.
\end{align}
Alice carries out one of two possible measurements labelled by $m\in\{0,1\}$, their outcomes labelled by $b_2\in\{0,1\}$. On obtaining outcome $b_2$ for measurement $m$, Alice inputs the two-bit string $x=mb_2$ to the channel. Bob possesses one of two possible binary measurements labelled by $v\in\{0,1\}$ (their outcomes labelled by $z\in\{0,1\}$) and must use the output $y$ from the classical channel (which gives him some information about the possible inputs $X_y$) to decide his measurement strategy in order to infer the message Alice's message $m$. The full strategy is detailed below:

	\begin{align}
	&v\in\{0,1\},\nonumber\\
	&p(v|y)=\begin{cases}
	\delta_{v,1}, \textrm{ for }y=(2,b_2),\\
	\delta_{v,0},\textrm{ for }y=(P,b_1\oplus b_2),\\
	\textrm{arbitrary for } y=(1,b_1).
	\end{cases}\nonumber\\
	&z\in\{0,1\}\nonumber\\
	&g(z,y)=\begin{cases}
	b,\textrm{ for } y=(t,b)=(1,b_1),\\
	b\oplus z,\textrm{ for } y=(t,b)\in\{(2,b_2),(P,b_1\oplus b_2)\}
	\end{cases}\\
	&p(m'=m|z,y)=\delta_{g(z,y),m}.
	\end{align}
	
Assuming $p(m)=\frac{1}{2}$, the expression for the success probability is given by

\begin{align}\label{prevedelexpr}
S=\frac{1}{3}+\frac{1}{6}\sum_{b_2,z,m,v}p(b_2,z|m,v)\delta_{b_2\oplus z,mv}.
\end{align}
	
We refer to Appendix \ref{prevedelsuccessprobability} for a complete derivation of the above expression. Now, the joint statistics $p(b_2,z|m,v)$ can be interpreted as arising from a Bell-CHSH scenario \cite{CHSH}, i.e., a Bell scenario where each party has two possible binary-outcome measurements ($m$ and $v$ in this case), noting that the statistics is non-signalling. We can then define the success probability in such a CHSH game \cite{BCP14} as 
\begin{align}
S_{\rm CHSH}\equiv\frac{1}{4}\sum_{b_2,z,m,v}p(b_2,z|m,v)\delta_{b_2\oplus z,mv},
\end{align}
so that 
\begin{align}\label{prevedelchsh}
S=\frac{1}{3}+\frac{2}{3}S_{\rm CHSH}.
\end{align}

We have
\begin{enumerate}
	\item Classically: \begin{equation}
	\sum_{b_2,z,m,v}p(b_2,z|m,v)\delta_{b_2\oplus z,mv}\leq 3,
	\end{equation}
	so that $S=S_{\rm Cl}\leq \frac{5}{6}\approx 0.833$, and
	\item Quantumly: 
	\begin{equation}
	\sum_{b_2,z,m,v}p(b_2,z|m,v)\delta_{b_2\oplus z,mv}\leq 2+\sqrt{2},
	\end{equation}
	so that $S=S_{\rm Q}\leq \frac{1}{3}+\frac{2+\sqrt{2}}{6}\approx 0.902$.
\end{enumerate}	
And, of course, on allowing arbitrary nonsignalling correlations, a PR-box \cite{PR94, BCP14} can achieve $S=S_{\rm PR}=1$.

How is the CHSH game related to the nonlocal game corresponding to the Prevedel {\em et al.}~protocol that one would obtain following Theorem \ref{isomorphism}? In Appendix \ref{prevedelgame}, we show that, in fact, this nonlocal game is essentially the CHSH game rewritten in such a way that Alice has two inputs while Bob has six.

\section{One-shot success probability of classical communication via general Kochen-Specker (KS) channels}\label{sec7}
\subsection{Channel hypergraph, encoding, and context-independent guessing (CIG):  the Cubitt {\em et al.}~strategy}
In this section we will go beyond channels with complete confusability graphs (for which one-shot zero-error communication is impossible) and consider general Kochen-Specker (KS) channels. Of particular interest will be KS channels that admit KS sets: for some of these channels it is possible to achieve an enhancement of the one-shot zero-error capacity using entanglement, e.g., in Ref.~\cite{CLM10}, Cubitt {\em et al.}~showed that one can use a classical channel based on Peres's $24$-ray two-qubit KS set \cite{ Peres91} which also underlies the Peres-Mermin proof of KS-contextuality \cite{Peres90, Mermin93}.

We will focus on the one-shot success probability that can be achieved using the Cubitt {\em et al.} \cite{CLM10} strategy when Bob only assumes the structure of the channel hypergraph, $H(\mathcal{N})$, and his knowledge of the encoding Alice uses but makes no assumptions about the exact channel probabilities $\{\mathcal{N}(y|x)\}_{x\in X,y\in Y}$. This could, for example, happen when Alice and Bob trust the channel hypergraph but they do not trust the channel probabilities of the classical channel given by some provider.

Recall that the one-shot success probability following the Cubitt {\em et al.} strategy is given by
\begin{align}
S=&\sum_mp(m)\sum_xp(x|m)\sum_y\mathcal{N}(y|x)\nonumber\\
&\sum_{m'}p(m'|y,m,x)\delta_{m,m'}\nonumber\\
=&\sum_mp(m)\sum_xp(x|m)\sum_y\mathcal{N}(y|x)\nonumber\\
&\sum_{x'\in X^{(m)}}p(x'|y,m,x).
\end{align}

Bob uses his knowledge of the channel hypergraph, $H(\mathcal{N})$, and the output received from the channel, $y$, along with any nonsignalling correlations shared with Alice, to make a guess $x'\in X_y$ for Alice's input $x$. Our assumption that Bob is oblivious of the channel probabilities means that the probability with which Bob makes his guess, $p(x'|y,m,x)$, should be independent of the particular $y\in Y_{x'}$ (``context" of the guess $x'$) that Bob receives. It can depend only on the support of $x'$ via an indicator function, i.e., $p(x'|y,m,x)=p(x'|m,x)\delta(y\in Y_{x'})$, where $\delta(y\in Y_{x'})=1$ if $y\in Y_{x'}$ and $0$ otherwise. Overall, we have
\begin{align}\label{cig}
&\textrm{for any } x'\in X:\nonumber\\
&p(x'|y_1,m,x)=p(x'|y_2,m,x)\equiv p(x'|m,x),\nonumber\\
&\forall y_1,y_2\in Y_{x'}, \forall x\in X^{(m)}, \forall m\in {\rm \bf Msg}.
\end{align}
We term this condition {\em context-independent guessing} (CIG). We will see that the quantum strategy of the Cubitt {\em et al.} protocol \cite{CLM10} satisfies this constraint and this fact allows us to invoke the assumption of measurement noncontextuality in addition to preparation noncontextuality in placing a noncontextual upper bound on the one-shot success probability. This will in turn allow us to analyze the Cubitt {\em et al.}~construction and the critical role of the KS theorem in it in the light of generalized noncontextuality \`a la Spekkens \cite{Spekkens05}. More concretely, we will see that a non-trivial upper bound on the classical one-shot success probability in this protocol can be characterized by a hypergraph invariant -- the weighted max-predictability \cite{KS18} -- following the approach of Ref.~\cite{Kunjwal20}.

On the other hand, note that classically we have $p(x'|y,x,m)=p(x'|y)$ for all $x\in X^{(m)}, m\in{\rm\bf Msg}$ and the CIG constraint of Eq.~\eqref{cig} (in a classical strategy) then requires that $p(x'|y)=p(x')\delta(y\in Y_{x'})$. That is,
\begin{align}\label{CCcig}
\textrm{for any } x'\in X&: p(x'|y_1,x,m)=p(x'|y_2,x,m)\equiv p(x'),\nonumber\\
&\forall y_1,y_2\in Y_{x'}, \forall x\in X^{(m)}, \forall m\in {\rm \bf Msg}.
\end{align}
Indeed, in a classical strategy for this communication task, any assignment $p:X\rightarrow [0,1]$ respecting the context-independence property defines what is usually called a (general) probabilistic model when the channel hypergraph is viewed as a contextuality scenario \cite{AFL15, Kunjwal20}.

\subsection{One-shot success probability of a KS channel under the CIG constraint: contextuality and  quantum advantage}
\begin{figure*}
	\centering
	\includegraphics[scale=0.15]{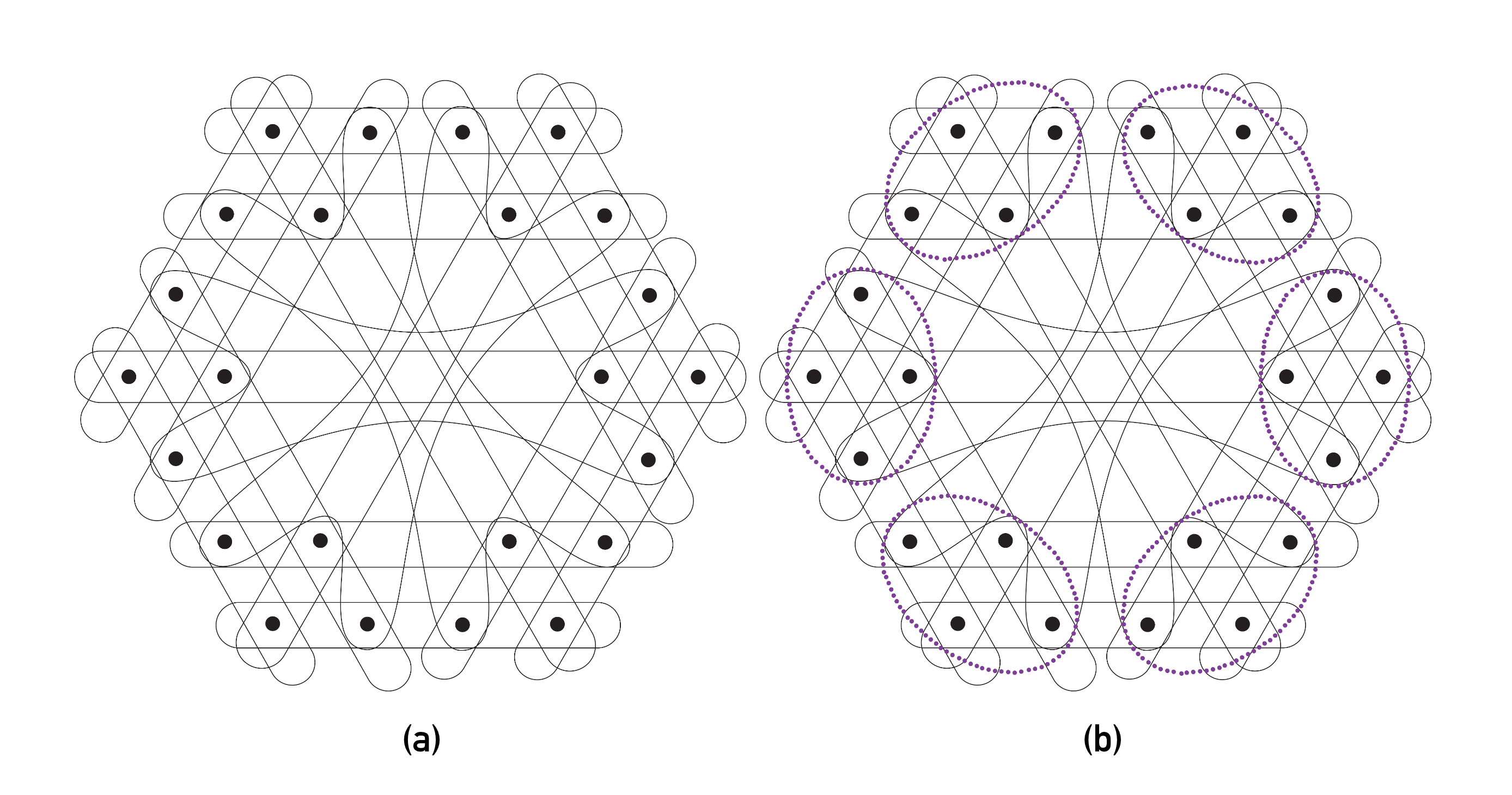}
	\caption{(a) The channel hypergraph considered in Ref.~\cite{CLM10}, and (b) the encoding used, highlighed with dotted hyperedges. Each message,  corresponding to a dotted loop, contains four mutually confusable input symbols.}
	\label{kschannel}
\end{figure*}
As our working example, we will consider the same classical channel considered by Cubitt {\em et al.}, cf.~Fig.~\ref{kschannel}. This channel admits a KS set, i.e., a set of projectors on a $4$-dimensional quantum system, each projector associated with a vertex in the channel hypergraph such that each hyperedge constitutes a projective measurement.
\subsubsection{The one-shot success probability with context-independent guessing}
The one-shot success probability is given by 
	\begin{align}
	S=&\sum_mp(m)\sum_xp(x|m)\nonumber\\
	&\sum_y\mathcal{N}(y|x)
	\sum_{x'\in X^{(m)}}p(x'|y,m,x)\nonumber\\
	=&\sum_mp(m)\sum_xp(x|m)\nonumber\\
	&\sum_y\Bigg(\mathcal{N}(y|x)\sum_{x'\in X^{(m)}}p(x'|m,x)\delta(y\in Y_{x'})\Bigg)\nonumber\\
	&\textrm{(using the CIG constraint)}\nonumber\\
	=&\sum_mp(m)\sum_xp(x|m)\nonumber\\
	&\sum_{x'\in X^{(m)}}\Bigg(\sum_y\mathcal{N}(y|x)\delta(y\in Y_{x'})p(x'|m,x)\Bigg).
	\end{align}
For any two symbols $x,x'\in X$, we quantify the confusability of $x$ with respect to $x'$ via the function
\begin{equation}
\eta(x,x')\equiv \sum_y\mathcal{N}(y|x)\delta(y\in Y_{x'}).
\end{equation}
This is the probability that, for input $x$, the channel $\mathcal{N}$ yields an output that could also arise from the input $x'$. Obviously, $\eta(x,x)=1$ for all $x\in X$.

Noting that $p(x|m)p(x'|m,x)=p(x,x'|m)$, we have 
\begin{align}
S&=\sum_mp(m)\sum_{x,x'\in X^{(m)}}p(x,x'|m)\eta(x,x')\label{oneshotcig}\\
&=\sum_mp(m)\sum_{x,x'\in X^{(m)}}\delta_{x,x'}p(x,x'|m)\nonumber\\
&+\sum_mp(m)\sum_{x,x'\in X^{(m)}}(1-\delta_{x,x'})p(x,x'|m)\eta(x,x'),
\end{align}
where we used the fact that $\eta(x,x)=1$. 
Defining
\begin{align}
S_{\rm perf}&\equiv\sum_mp(m)\sum_{x,x'\in X^{(m)}}\delta_{x,x'}p(x,x'|m),\\
S_{\rm imperf}&\equiv\sum_mp(m)\sum_{x,x'\in X^{(m)}}(1-\delta_{x,x'})p(x,x'|m)\eta(x,x'),
\end{align}
we have that $S=S_{\rm perf}+S_{\rm imperf}$. Here $S_{\rm perf}$ denotes the contribution to the success probability from the situation where Bob guesses Alice's input $x$ to the channel exactly (i.e., $x'=x$, and therefore also infers $m$ correctly) and $S_{\rm imperf}$ denotes the remaining contribution to the success probability from the situation where Bob doesn't guess $x$ correctly (i.e., $x'\neq x$) but nevertheless infers $m$ correctly from $x'$ (i.e., $x'\in X^{(m)}$).

Recalling the source-measurement correlation function ${\rm Corr}$ studied in Ref.~\cite{Kunjwal20}, we have
\begin{equation}\label{defncorr}
{\rm Corr}=S_{\rm perf}\equiv \sum_mp(m)\sum_{x,x'\in X^{(m)}}\delta_{x,x'}p(x,x'|m).
\end{equation}
In the context of Ref.~\cite{Kunjwal20}, this correlation function captures how predictable the measurements corresponding to the hyperedges $m$ can be made when one varies over corresponding preparation ensembles (also labelled by $m$) that average to the same mixed state. In the ideal case of projective measurements from Peres' 24-vector KS set, this quantity is $1$, since every measurement is perfectly predictable when the input state is picked from its orthonormal basis and the uniform mixture over input states from any such basis is the maximally mixed state. However, in a noncontextual ontological model, this quantity is a upper bounded by a hypergraph invariant \cite{Kunjwal20} that will turn out to be relevant for the one-shot success probability in the following subsections.

We then have the following bounds on $S$:
\begin{align}\label{sfncorr}
&{\rm Corr}+\eta_{\rm min}(1-{\rm Corr})\nonumber\\
&\leq S={\rm Corr}\nonumber\\
&+\sum_mp(m)\sum_{x,x'\in X^{(m)}}(1-\delta_{x,x'})p(x,x'|m)\eta(x,x')\nonumber\\
&\leq {\rm Corr}+\eta_{\rm max}(1-{\rm Corr}),
\end{align}
where 
\begin{align}
&\eta_{\rm min}\equiv\min_m\min_{x\neq x'\in X^{(m)}}\eta(x,x'),\\
&\eta_{\rm max}\equiv\max_m\max_{x\neq x'\in X^{(m)}}\eta(x,x').
\end{align}
Here $\eta_{\min}$ is the minimum confusability between any two symbols in any  message in the encoding. Similarly, $\eta_{\max}$ is the maximum confusability between any two symbols in any message in the encoding. We have $0<\eta_{\min}\leq\eta_{\max}<1$.

\subsubsection{The assumption of measurement noncontextuality and how any noncontextual strategy satisfies context-independent guessing}
In the Cubitt {\em et al.} strategy, it is assumed that both Alice and Bob have access to specific sets of measurements carved out of a KS set for the channel of Fig.~\ref{kschannel}. Alice's six measurements $\{\mathbb{M}^{\rm (A)}_m\}_{m=1}^6$ correspond to the six dotted hyperedges in Fig.~\ref{kschannel}, while Bob's measurements $\{\mathbb{M}^{\rm (B)}_y\}_{y=1}^{18}$ correspond to partitioning the $24$ projectors into $18$ hyperedges, denoted by solid hyperedges in Fig.~\ref{kschannel}.

In our treatment of the problem, since we want to allow more general choices of measurements, we make two relaxations: 
\begin{enumerate}
	\item Firstly, we do not restrict ourselves to projective measurements on Bob's side, i.e., we allow any set of positive operators (each operator associated to a vertex in the channel hypergraph) that satisfy the requirement that the (solid {\em and} dotted) hyperedges in Fig.~\ref{kschannel} form complete measurements, and
	\item Secondly, although in the Cubitt {\em et al.}~protocol, Alice's measurements are carved out of the same set of positive operators that Bob can implement, and this is clearly the optimal choice of measurements for Alice (yielding $S=1$ for sending six messages, quantumly), we allow that, in general, Alice could associate some {\em other} measurements with the messages she wants to send even if the outcomes of such measurements do not satisfy the operational equivalences implicit in the channel hypergraph of Fig.~\ref{kschannel}. That is, the encoding strategy of Alice could use a set of positive operators for her measurements $\{\mathbb{M}^{\rm (A)}_m\}_{m=1}^6$ that is completely different from the set of positive operators used by Bob for his measurements $\{\mathbb{M}^{\rm (B)}_y\}_{y=1}^{18}$ in the decoding strategy.
\end{enumerate}  

 We mention this to emphasize that our generalization of the Cubitt {\em et al.}~strategy does not rely on identifying the measurement outcomes of Alice and Bob in the way they are identified in the optimal strategy and our use of the assumption of measurement noncontextuality is restricted to Bob's system, i.e., response functions of Bob's measurements. In particular, the positive operators that constitute Bob's measurements can, in principle, be reconfigured to define the six measurements $\{\mathbb{M}^{\rm (B)}_m\}_{m=1}^6$ that it would be optimal for Alice to choose for her encoding measurements $\{\mathbb{M}^{\rm (A)}_m\}_{m=1}^6$. 

In a noncontextual ontological model of Bob's system, the response functions associated with the vertices (labelled by $x\in X$) respect measurement noncontextuality, i.e., for all $x,y_1,y_2$ such that $x\in X_{y_1}\cap X_{y_2}\neq \varnothing$,
\begin{equation}
\xi(x|y_1,\lambda)=\xi(x|y_2,\lambda)\equiv \xi(x|\lambda),\forall\lambda\in\Lambda.
\end{equation}

Now, even though Bob may not implement or have access to the six measurements $\{\mathbb{M}^{\rm (B)}_m\}_{m=1}^6$ that would be optimal for Alice in the protocol, the operational equivalences implicit in the hypergraph of Fig.~\ref{kschannel} indicate that the response functions for $\mathbb{M}^{\rm (B)}_m$ on Bob's system must, under the assumption of measurement noncontextuality, also satisfy
\begin{equation}
\xi(x|m,\lambda)=\xi(x|\lambda),\forall\lambda\in\Lambda,
\end{equation}
for all $x\in X$ and all $m\in {\rm \bf Msg}$. 

Recalling that 
\begin{equation}
	p(x'|y,m,x)=\sum_{\lambda}\xi(x'|y,\lambda)\mu(\lambda|m,x),
\end{equation}
we have that $p(x'|y_1,m,x)=p(x'|y_2,m,x)$ for all $x',y_1,y_2,m,x$ such that $x'\in X_{y_1}\cap X_{y_2}$ (equivalently, $y_1,y_2\in Y_{x'}$), so that the CIG constraint is satisfied by any noncontextual strategy and the one-shot success probability takes the form of Eq.~\eqref{oneshotcig}.

\subsubsection{Upper bound on the one-shot success probability from preparation and measurement noncontextuality}
We now proceed to upper bound the success probability under the assumption of noncontextuality, i.e., preparation and measurement noncontextuality, and obtain
\begin{align}
S\leq &\max_{\lambda}\sum_mp(m)\sum_{x,x'\in X^{(m)}}\xi(x'|m,\lambda)\mu(x|m,\lambda)\eta(x,x')\nonumber\\
\equiv &S^{\rm max}_{\rm NC}.
\end{align}
To see how this comes about starting from the expression for the one-shot success probability in Eq.~\eqref{oneshotcig}, i.e.,
\begin{align}
S=\sum_mp(m)\sum_{x,x'\in X^{(m)}}p(x,x'|m)\eta(x,x'),\nonumber
\end{align}
first note that $p(x,x'|m)=\sum_{\lambda}\xi(x'|m,\lambda)\mu(\lambda,x|m)$, where $\mu(\lambda,x|m)=\mu(\lambda|m,x)p(x|m)=\mu(x|m,\lambda)\mu(\lambda|m)$. As in the general case of one-shot classical communication assisted by entanglement, given the operational equivalence of ensembles prepared on Bob's side by Alice's measurements $\mathbb{M}^{\rm (A)}_m$ and the assumption of preparation noncontextuality, we have that $\mu(\lambda|m)=\nu(\lambda)$ 
for all $m$. We then have
\begin{align}\label{sontlmodel}
S
=&\sum_mp(m)\nonumber\\
&\sum_{x,x'\in X^{(m)}}\sum_{\lambda}\xi(x'|m,\lambda)\mu(x|m,\lambda)\nu(\lambda)\eta(x,x')\nonumber\\
=&\sum_{\lambda}\nu(\lambda)\sum_mp(m)\nonumber\\
&\sum_{x,x'\in X^{(m)}}\xi(x'|m,\lambda)\mu(x|m,\lambda)\eta(x,x')\nonumber\\
\leq&
S^{\max}_{\rm NC}.
\end{align}

\subsubsection{Contextuality drives the quantum advantage}
We show that the one-shot success probability achievable via any noncontextual strategy is no better than best classical strategy with context-independent guessing. For any extremal classical strategy $i\in \mathcal{I}$ ($\mathcal{I}$ being the set of extremal classical strategies satisfying CIG), the success probability is given by 
\begin{align}
&S_{\rm Cl}(i)\nonumber\\
=&\sum_mp(m)\sum_xp_{\rm A}(x|m,i)\sum_y\mathcal{N}(y|x)\sum_{x'\in X^{(m)}}p_{\rm B}(x'|y,i)\nonumber\\
=&\sum_mp(m)\sum_xp_{\rm A}(x|m,i)\sum_y\mathcal{N}(y|x)\nonumber\\
&\sum_{x'\in X^{(m)}}p_{\rm B}(x'|i)\delta(x'\in X_y),
\end{align}
where $p_{\rm B}(x'|y,m,x,i)=p_{\rm B}(x'|y,i)=p_{\rm B}(x'|i)\delta(x'\in X_y)$, since Bob has no access to any information from Alice besides the shared variable $i$ (denoting the strategy both of them agree to implement) and the channel output $y$, the latter specifying a confusable set $X_y$ containing Alice's input $x$ (and Bob's guess $x'$). An arbitrary classical strategy can then be represented by a convex mixture of extremal classical strategies according to some probability distribution $\{p(i)\}_{i\in\mathcal{I}}$ and the classical success probability is then given by

\begin{align}
&S_{\rm Cl (CIG)}=\sum_{i} p(i)S_{\rm Cl(CIG)}(i)
\leq\max_{i}S_{\rm Cl(CIG)}(i)\nonumber\\
=&\max_{i}\sum_mp(m)\sum_xp_{\rm A}(x|m,i)\sum_y\mathcal{N}(y|x)\nonumber\\
&\sum_{x'\in X^{(m)}}p_{\rm B}(x'|y,i)\nonumber\\
\equiv& S^{\max}_{\rm Cl(CIG)}.
\end{align}
Using the CIG constraint, we have $p_{\rm B}(x'|y,i)=p_{\rm B}(x'|i)\delta(y\in Y_{x'})$, where any extremal classical strategy $i$ specifies a particular extremal probabilistic model on the channel hypergraph (viewed as a contextuality scenario \cite{AFL15}) in Fig.~\ref{kschannel}(a). This allows us to express the maximal classical success probability satisfying CIG as
\begin{align}
&S^{\max}_{\rm Cl(CIG)}\nonumber\\
=&\max_{i}\sum_mp(m)\sum_xp_{\rm A}(x|m,i)\sum_{x'\in X^{(m)}}\Bigg(p_{\rm B}(x'|i)\nonumber\\
&\sum_y\mathcal{N}(y|x)\delta(y\in Y_{x'})\Bigg)\nonumber\\
=&\max_{i}\sum_mp(m)\sum_{x,x'\in X^{(m)}}p_{\rm A}(x|m,i)p_{\rm B}(x'|i)\eta(x,x')\nonumber\\
\end{align}
We therefore have
\begin{align}
S^{\max}_{\rm NC}=S^{\max}_{\rm Cl(CIG)}.
\end{align}
Since any classical strategy is a convex mixture of extremal classical strategies, the upper bound $S^{\max}_{\rm Cl(CIG)}$ can always be achieved by a classical strategy, i.e., there exists an extremal classical strategy $i_*\in\mathcal{I}$ such that $S_{\rm Cl(CIG)}(i_*)=S^{\max}_{\rm Cl(CIG)}$. Similarly, the upper bound $S^{\max}_{\rm NC}$ can be saturated by a noncontextual strategy, albeit a very trivial one, following a similar reasoning as at the end of Section \ref{prepctxnecc} (except that the extremal response functions here are indeterministic on account of KS-uncolourability). Thus, we have that contextuality also drives the quantum advantage in one-shot classical communication when Alice and Bob trust the channel hypergraph but make no assumptions about the channel probabilities.\footnote{Note that, under the CIG constraint, we no longer have the exact correspondence with nonlocal games exemplified by Theorem \ref{isomorphism}. This is because the connection between prepare-and-measure scenarios on Bob's system alone and Bell scenarios where Bob is one of the parties in a two-party Bell experiment breaks down when one imposes, besides preparation noncontextuality, the assumption of measurement noncontextuality on Bob's system (based on the operational equivalences between his local measurement events). This restricts the scope of response functions for Bob's measurements beyond anything required by local causality (under which no restriction on local response functions is imposed). The interested reader may look at a discussion of this point at the end of Section $2.7$ in Ref.~\cite{Kunjwal19}.}

\subsubsection{Contextuality witnessed by a hypergraph-invariant -- the weighted max-predictability -- is sufficient for a quantum advantage}

In this section we point out an explicit connection between contextuality that can be witnessed via the noise-robust noncontextuality inequalities of Ref.~\cite{Kunjwal20} and the quantum advantage in the one-shot classical communication task with context-independent guessing. The full technical argument supporting the claims in this section is presented in Appendix \ref{hypergraphinvariant}.

In Appendix \ref{hypergraphinvariant}, we first show that 
\begin{align}
	S\leq S^{\max}_{\rm NC}\leq \eta_{\max}+(1-\eta_{\max})\beta(\Gamma,\{p(m)\}_m),
\end{align}
where 
\begin{align}
\beta(\Gamma,\{p(m)\}_m)\equiv\max_{\lambda} \sum_mp(m)\max_{x\in X^{(m)}}\xi(x|m,\lambda)
\end{align}
is the weighted max-predictability, a hypergraph invariant that was defined in Ref.~\cite{Kunjwal19} and studied extensively in Ref.~\cite{Kunjwal20}, appearing in the upper bounds of noise-robust noncontextuality inequalities proposed in the frameworks of Refs.~\cite{Kunjwal19,Kunjwal20}. Here, $\Gamma$ is the hypergraph defined by the contextuality scenario from which the measurements of Alice and Bob are drawn in the ideal Cubitt {\em et al.}~strategy \cite{CLM10}, i.e., Fig.~\ref{kschannel}(b) including all the hyperedges (solid and dotted).

We then consider the special case $\eta_{\max}=\eta_{\min}=\eta$ and show that firstly, from Eq.~\eqref{sfncorr}, we have 

\begin{align}
S&={\rm Corr}+\eta (1-{\rm Corr})=\eta+{\rm Corr}(1-\eta).\label{etacorrbound}
\end{align}
We then have that
\begin{align}
{\rm Corr}&>\beta(\Gamma,\{p(m)\}_m)\nonumber\\
\Leftrightarrow S&>\eta+(1-\eta)\beta(\Gamma,\{p(m)\}_m)\nonumber\\
&\geq S^{\max}_{\rm NC}=S^{\max}_{\rm Cl(CIG)}\nonumber\\
\Rightarrow S&>S^{\max}_{\rm NC}.
\end{align}

Now, 
\begin{align}\label{ncineq1}
{\rm Corr}\leq \beta(\Gamma,\{p(m)\}_m)
\end{align}
is an instance of a noise-robust noncontextuality inequality following the approach of Ref.~\cite{Kunjwal20}, inspired by logical proofs of KS-contextuality \cite{KS67,KS15}. Hence, the contextuality witnessed by ${\rm Corr}>\beta(\Gamma,\{p(m)\}_m)$ is sufficient for a quantum advantage in this task when $\eta_{\max}=\eta_{\min}$.\footnote{Note that 
when $\eta_{\min}=\eta_{\max}\equiv \eta$, all pairs of distinct input symbols are equally confusable, i.e., $\eta(x,x')=\sum_y\mathcal{N}(y|x)\delta(y\in Y_{x'})$ is constant across all $x,x'\in X$.} Output-uniform channels with $k$-regular hypergraphs have $\eta(x,x')=\frac{1}{k}\sum_y\delta(y\in Y_x\cap Y_{x'})$ and this quantity is independent of $x,x'$ ($x\neq x'$) if and only if $|Y_x\cap Y_{x'}|$ is constant for all $x\neq x'$, i.e., the number of hyperedges shared by any two confusable vertices of the channel hypergraph is constant across all pairs of confusable vertices. In the channel hypergraph of Fig.~\ref{kschannel}, for example, we have $|Y_x\cap Y_{x'}|=1$ for all confusable pairs of vertices $x,x'$, so the classical channel studied in Ref.~\cite{CLM10} (where $k=3$ and $\mathcal{N}(y|x)=\frac{1}{3}$ for all 
$y\in Y_x, x\in X$) satisfies the condition required for ${\rm Corr}>\beta(\Gamma,\{p(m)\}_m)$ to imply a quantum advantage. The case ${\rm Corr}=1$ corresponds to the situation studied in Ref.~\cite{CLM10}. We have shown that this quantum advantage can persist even when ${\rm Corr}<1$, i.e., in the regime of noisy measurements.

Another special case concerns the situation where $S_{\rm imperf}=0$. In this situation too, the violation of ${\rm Corr}\leq \beta(\Gamma,\{p(m)\}_m)$ implies a quantum advantage (cf.~Appendix \ref{hypergraphinvariant}). Using the set-up in Ref.~\cite{CLM10}, one can achieve ${\rm Corr}=1$, i.e., zero-error communication.

We have thus provided an instance of an information-theoretic task where the noise-robust signatures of contextuality \`a la Refs.~\cite{KS15,Kunjwal20} witness a quantum advantage in the task. This provides an operational meaning to noise-robust noncontextuality inequalities of the type in Eq.~\eqref{ncineq1} that were proposed in Ref.~\cite{Kunjwal20}.

\subsection{KS basis sets and the ideal Cubitt {\em et al.}~strategy}
Can one use the strategy of Ref.~\cite{CLM10} starting from {\em any} KS set of vectors?\footnote{Recall that the general protocol of Section \ref{sec3_1} does not rely on the existence of KS sets. It's only the particular strategy of Ref.~\cite{CLM10} that makes use of them.} The strategy requires not merely a KS set -- namely, a set of vectors with orthogonality relations represented by a KS-uncolourable hypergraph -- but, in fact, a {\em KS basis set}, i.e., a set of \textit{disjoint} complete orthogonal bases $\mathcal{Z}\equiv\{\mathcal{B}_m\}_{m=1}^q$ such that
it is impossible to pick a vector from each basis ensuring that no two are orthogonal. Clearly, the vectors in a KS basis set constitute a KS set. Denoting the vectors in basis $\mathcal{B}_m$ as $\{\psi_{mj}\}_{j=1}^d$, where $d$ is the dimension of the Hilbert space spanned by the basis, we construct a classical channel $\mathcal{N}$ with inputs labelled by $\{(m,j)|m\in [q], j\in [d]\}$ ($[N]$ denoting $\{1,2,\dots,N\}$ for any positive integer $N$). The confusability graph, $G(\mathcal{N})$, of the channel is such that two inputs are confusable if and only if the corresponding vectors are orthogonal. The definition of a KS basis set then implies that $\alpha(G(\mathcal{N}))<q$ for any channel thus constructed from it. 

As we have noted, the set of vectors appearing in any KS basis set form a KS set. However, it is not {\em a priori} obvious that, given a KS set, it is always possible to carve it up into a KS basis set $\{\mathcal{B}_m\}_{m=1}^q$ such that $q>\alpha(G(\mathcal{N}))$ for any channel constructed from it following the prescription of Cubitt {\em et al.} \cite{CLM10} mentioned above.\footnote{In Ref.~\cite{CLM10}, there is a claim that the existence of KS basis sets is ``a corollary of the KS theorem". This is true if a KS basis set is allowed, in general, to contain bases that share vectors. However, for the Cubitt {\em et al.}~construction to work, the bases in a {\em KS basis set} must be disjoint, i.e., no vectors are shared between bases: this is what allows Alice to encode her messages unambiguously in the outcomes of these measurement bases. Further, for an advantage, the number of these disjoint bases in a KS basis set must exceed the independence number of the confusability graph of the channel constructed from the KS basis set. Hence, in our definition of a {\em KS basis set}, we explicitly include the disjointness of bases, something Ref.~\cite{CLM10} implicitly assumed.} Hence, to answer whether the entanglement-assisted enhancement of the one-shot zero-error capacity achieved by the strategy of Ref.~\cite{CLM10} carries through for any KS set, we need to settle the following question:

{\em Does every KS set admit a KS basis set of size $q>\alpha(O(\Gamma))$? Here $\Gamma$ is the contextuality scenario corresponding to the KS set and $O(\Gamma)$ is the orthogonality graph of $\Gamma$.\cite{Kunjwal19,Kunjwal20}}

If the answer is in the affirmative, then the Cubitt {\em et al.}~strategy can be used starting from arbitary KS sets. If not, then there must exist a counter-example. Indeed, we can find such a counter-example and, therefore, the Cubitt {\em et al.}~strategy is not applicable to arbitrary KS sets: it only works for KS sets that admit (disjoint) KS basis sets. Our counter-example comes from the Conway-Kochen $31$-vector KS set, the smallest known KS set in dimension $d=3$ \cite{Peres06}. We refer to Appendix \ref{conway31} for details.

This raises the following important open problem: {\em Given an arbitrary KS set, what are the necessary and sufficient criteria for it to admit a KS basis set?}

\section{Discussion and outlook}\label{sec8}

We have generalized and unified the protocols of Refs.~\cite{CLM10, PLM11} in a broad framework for entanglement-assisted one-shot classical communication that should prove useful for future investigations. Our results bear witness to the role that noise-robust contextuality \`a la Spekkens \cite{Spekkens05} plays in this task. Indeed, the problem of entanglement-assisted one-shot classical communication provides a fertile ground to study the rich interplay between the Kochen-Specker theorem \cite{KS67}, Spekkens contextuality \cite{Spekkens05} and its hypergraph-theoretic formulations \cite{Kunjwal19,Kunjwal20}, and nonlocal games.  Several open questions and opportunities for future work arise:

\begin{enumerate}
	
	\item Does there exist a generic construction of a nonlocal game, instantiating Proposition \ref{nonlocalgamexistence}, for any channel $\mathcal{N}$ that admits an enhancement of its one-shot success probability? Even extending the family of channels for which such a construction exists beyond the case we have shown, i.e., the family of output-uniform $k$-regular channels of Theorem \ref{isomorphism}, would constitute progress in this direction. Furthermore, even within this family of channels, an important question is to characterize those for which the corresponding nonlocal game admits a gap between classical and quantum/nonsignalling correlations. These channels would, in turn, admit an advantage in a corresponding one-shot communication task because of Theorem \ref{isomorphism}.
	
	\item While the channels considered in Refs.~\cite{CLM10, PLM11} are KS channels, it remains an open question whether more general channels (in particular, with KS-colourable channel hypergraphs) exhibit non-trivial advantages in enhancing the one-shot success probability using entanglement. Our general protocol in Section \ref{sec3} does not specifically rely on the channel being KS-uncolourable. For example, a simple channel corresponding to a statistical proof of KS-contextuality is the one based on the KCBS construction on a qutrit \cite{KCBS}. It consists of $10$ vertices, denoted $\{v_i,w_i\}_{i=1}^5$, and $5$ hyperedges, denoted $\{v_i,w_i,v_{i+1}\}_{i=1}^5$ (addition modulo $5$, so that $i+1=1$ for $i=5$), with $\alpha(G(\mathcal{N}))=3$ (equal to its one-shot zero-error capacity); a natural question then arises: can entanglement be used to enhance the one-shot sucess probability of sending two bits using this channel? What would be the role of noise-robust contextuality \`a la Refs.~\cite{KS18,Kunjwal19} in enabling such an enhancement? If not, can {\em any} channel with a KS-colourable hypergraph admit enhancement of its one-shot success probability?
	
	\item The Cubitt {\em et al.} protocol \cite{CLM10} requires the existence of (disjoint) KS basis sets. Is it possible to modify this protocol to use KS sets which do not admit (disjoint) KS basis sets, e.g., the Conway-Kochen 31-vector KS set? Would such a modification still allow for the possibility of enhancing the one-shot zero-error capacity of a classical channel? Or would it, maybe, only allow for an enhancement of the one-shot success probability following the general protocol we discussed in Section \ref{sec3}? The existence of pseudotelepathy games based on KS sets \cite{BBT05} suggests that a quantum advantage in some corresponding one-shot communication task (following Theorem \ref{isomorphism}) should be possible. A related question is: what is the simplest scenario that admits enhancement of the one-shot zero-error capacity of a classical channel? Is the example studied in Ref.~\cite{CLM10} the simplest one, or is it possible to further reduce, say, the size of the input alphabet or the dimension of the quantum system for which the enhancement is achieved? Of course, insofar as one uses KS sets to achieve this enhancement, this is also related to the smallest possible KS sets: in dimension $3$, it's been shown that the smallest KS set can have no fewer than $22$ vectors \cite{UW16, Arends09} (the $31$-vector Conway-Kochen construction still being the smallest one known in $3$ dimensions). Is there a smaller KS set (fewer than $24$ vectors) than Peres's 24-vector set that also admits a KS basis set?

\item The enhancement of the one-shot success probability using the Cubitt {\em et al.}~strategy relies on the orthogonality (compatibility) relations between projectors (projective measurements). Are there nontrivial examples of enhancement of the one-shot success probability that are only achievable with nonprojective measurements, perhaps inspired by joint measurability structures that lie outside the purview of projective measurements \cite{KHF14, AK20}? We know that there exist such joint measurability structures, e.g., Specker's scenario, admitting proofs of contextuality \cite{KG14, ZCL17}.

\item The connection of the one-shot communication task with preparation contextuality could also be leveraged to obtain bounds on inaccessible information in preparation contextual ontological models of quantum theory, following the ideas recently proposed in Ref.~\cite{Marvian20}.
\end{enumerate}

More generally, the problem of entanglement-assisted one-shot zero-error communication can be viewed as a channel simulation problem, i.e., using a noisy channel to simulate a noiseless channel in a one-shot setting using nonsignalling correlations \cite{CLM11}. The relaxation of it to the case of enhancing the one-shot success probability (which we have studied) can be viewed as using a noisy channel to simulate a less noisy channel using nonsignalling correlations, i.e., noise-attenuation of a classical channel using a nonclassical common-cause resource \cite{WSS20}. We have focussed in this paper on the interplay of this latter channel simulation problem with the contextuality of the system that the receiver (Bob) holds in the communication task. A worthwhile project here is a rigorous resource-theoretic account of this problem to better understand how various resources affect the simulation preorder over classical channels in this noisy setting \cite{Shannon58}: whether perhaps the resource of LOSR-entanglement \cite{SFK20} is more appropriate than LOCC-entanglement when viewing the resource aspects of entanglement (and how this affects, for example, the usefulness of Tsirelson boxes vs.~Hardy boxes \cite{Hardy93,WSS20,SFK20} in this task), how the resource of noise-robust contextuality on one wing of a Bell experiment plays with bipartite nonlocality, and, more abstractly, the usefulness of a common-cause resource in simulating a direct-cause resource (e.g., the fact that entanglement can increase the one-shot zero-error capacity of a classical channel \cite{CLM10}). It would also be interesting to see if the contextuality witnesses we have considered in this paper turn out to be related to some monotones for channel (non-)conversions in a resource theory of channel simulation.

\ack
S.A.Y. is supported by the Richardson Endowment in the Department of Physics, Duke University, and Melvin J. Rieger Scholarship Fund in Physics and Robert E. Boyer Endowed Presidential Scholarship for Natural Sciences in the Department of Physics, University of Texas at Austin. S.A.Y. was also supported by resources and people in and around the Kodosky Reading room, Department of Physics, University of Texas at Austin. S.A.Y. acknowledges Scott Aaronson (UT Austin) for mentoring related undergraduate thesis work, Iman Marvian and Austin Hulse (Duke) for useful discussions, and Saathwik Yadavalli for helping with figures. R.K.~is supported by the Charg\'e de Recherche fellowship of the Fonds de la Recherche Scientifique - FNRS (F.R.S.-FNRS), Belgium. R.K. would like to thank Matt Pusey for raising the issue of the role of nonlocality in the early phase of this project when some of these ideas were presented in a ``Lobster Lunch" talk at Oxford. We also thank Rob Spekkens for some early conversations on this topic when this project was just getting started. We acknowledge support from the Perimeter Institute for Theoretical Physics through their Undergraduate Summer Research Program, 2018, which led to this collaboration. Research at Perimeter Institute is supported by the Government of Canada through
the Department of Innovation, Science and Economic
Development Canada and by the Province of Ontario
through the Ministry of Research, Innovation and Science.

\bibliographystyle{apsrev4-2}
\nocite{apsrev41Control}
\bibliography{masterbibfile}

\begin{appendix}
\section{Preparation noncontextuality on one wing of a bipartite Bell experiment implies local causality}\label{pncimplieslc}
We will use quantum notation below for ease of understanding, but the argument applies to all non-signalling general probabilistic theories (GPTs). 

Consider a general bipartite Bell scenario where Alice's measurement settings are labelled by $s$, Bob's settings are labelled by $t$, and their respective outcomes are labelled by $a$ and $b$. Their joint statistics is, therefore, given by $p(a,b|s,t)=\Tr (E^{(s)}_a\otimes E^{(t)}_b\rho_{\rm AB})$, where $\{E^{(s)}_a\}_a$ denotes the POVM associated with $s$, $\{E^{(t)}_b\}_b$ denotes the POVM associated with $t$, and $\rho_{\rm AB}$ is the entangled state shared between Alice and Bob. We now consider the prepare-and-measure experiment on Bob's side\footnote{The same argument goes through with the roles of Alice and Bob interchanged.} that this Bell scenario induces: Bob's preparations are steered by Alice's measurements, i.e., every measurement outcome $E^{(s)}_a$ on Alice's side steers Bob's system to (an unnormalized state) $\sigma_{a|s}=\Tr_{\rm A}(E^{(s)}_a\otimes I\rho_{\rm AB})$. However, no-signalling requires that Bob should not be able to infer Alice's measurement setting $s$ by local interventions on his system alone, so that we have $\sum_a\sigma_{a|s}=\rho_{\rm B}$ for all measurement settings $s$ that Alice can choose. Each $s$ therefore labels a preparation ensemble $\{p(a|s),\rho_{a|s}\}_a$ on Bob's side such that $\sum_ap(a|s)\rho_{a|s}=\rho_{\rm B}$ for all $s$, where $p(a|s)=\Tr_B\sigma_{a|s}$ and $\rho_{a|s}=\frac{\sigma_{a|s}}{p(a|s)}$. Given this operational equivalence between the preparation ensembles on Bob's side, the assumption of preparation noncontextuality entails that any ontological model of Bob's system must satisfy $\sum_ap(a|s)p(\lambda|s,a)=p(\lambda)$ for all $s$. This can be rewritten as $\sum_ap(a|s,\lambda)p(\lambda|s)=p(\lambda)$ for all $s$, i.e., $p(\lambda|s)=p(\lambda)$ for all $s$. We then have, given Bob's measurement outcomes $E^{(t)}_b$, that
\begin{align}
p(b|t,s,a)=\Tr_B (E^{(t)}_b\rho_{a|s})=\sum_{\lambda}p(b|t,\lambda)p(\lambda|s,a),
\end{align}
and 
\begin{align}
&p(a,b|s,t)\nonumber\\
=&p(a|s)p(b|t,s,a)\nonumber\\
=&\sum_{\lambda}p(b|t,\lambda)p(a|s)p(\lambda|s,a)\nonumber\\
=&\sum_{\lambda}p(b|t,\lambda)p(a|s,\lambda)p(\lambda|s)\nonumber\\
=&\sum_{\lambda}p(b|t,\lambda)p(a|s,\lambda)p(\lambda),
\end{align}
where the last equality follows from the assumption of preparation noncontextuality. Thus, the existence of a preparation noncontextual ontological model for Bob's system (or for Alice's system, by symmetry) implies the existence of locally causal ontological model for the bipartite Bell experiment. Note that 
\begin{align}
&p(a,b|s,t)\nonumber\\
=&\Tr(E^{(s)}_a\otimes E^{(t)}_b\rho_{\rm AB})\nonumber\\
=&\Tr(E^{(s)}_a\otimes I \rho_{\rm AB})\Tr_{\rm B}(E^{(t)}_b\rho_{a|s}),
\end{align}
where $p(a|s)=\Tr(E^{(s)}_a\otimes I \rho_{\rm AB})$ and $p(b|t,s,a)=\Tr_{\rm B}(E^{(t)}_b\rho_{a|s})$.

\section{Proof of Theorem \ref{isomorphism}}\label{proofisomorphism}

	We begin by noting that, following Eq.~\eqref{localbound}, we have 
	\begin{align}
		&S_{\rm Cl}^{\max}\nonumber\\
		=&\max_{p(x,m'|m,y)\in\mathcal{L}}S\nonumber\\
		=&\frac{1}{|{\rm \bf Msg}|}\frac{1}{k}\max_{p(x,m'|m,y)\in\mathcal{L}}\sum_{m,m',x,y\in Y_x}\delta_{m,m'}p(x,m'|m,y).
	\end{align}
	The Bell expression of Eq.~\eqref{bellexp} can be rewritten by making the substitution (recalling that $p(y)=\frac{1}{|Y|}$)
	\begin{align}
		\sum_{y\notin Y_x}p(y)=1-\sum_{y\in Y_x}p(y)=1-\frac{k}{|Y|},
	\end{align}
	and using the no-signalling condition
	\begin{align}
		\sum_{m'}p(x,m'|m,y)=p(x|m),\textrm{ for all }y\in Y,
	\end{align}
	to express the second term of Eq.~\eqref{bellexp} as follows:
	\begin{align}
		&\frac{1}{|{\rm \bf Msg}|}\sum_{m,m',x,y\notin Y_x}\frac{1}{|Y|}p(x,m'|m,y)\nonumber\\
		=&\frac{1}{|{\rm \bf Msg}|}\sum_{m,x,y\notin Y_x}\frac{1}{|Y|}p(x|m)\nonumber\\
		=&\frac{1}{|{\rm \bf Msg}|}|{\rm \bf Msg}|\left(1-\frac{k}{|Y|}\right)\nonumber\\
		=&1-\frac{k}{|Y|}.
	\end{align}
	The Bell expression then becomes
	\begin{align}
		&S_{\rm Bell}\nonumber\\
		=&1-\frac{k}{|Y|}+\frac{1}{|{\rm \bf Msg}|}\frac{1}{|Y|}\sum_{m,m',x,y\in Y_x}\delta_{m,m'}p(x,m'|m,y)\nonumber\\
		=&1-\frac{k}{|Y|}+\frac{k}{|Y|}S,\label{bellpnc}
	\end{align}
	using Eq.~\eqref{oneshotexp}. This means that the following holds:
	\begin{equation}
		S=1\Leftrightarrow S_{\rm Bell}=1,
	\end{equation}
	i.e., the one-shot zero-error communication occurs if and only if the corresponding nonlocal game is won with certainty. On the other hand, this also means that 
	\begin{align}
		&S_{\rm local}^{\max}=\max_{p(x,m'|m,y)\in\mathcal{L}}S_{\rm Bell}\nonumber\\
		=&1-\frac{k}{|Y|}+\frac{k}{|Y|}\max_{p(x,m'|m,y)\in\mathcal{L}}S\nonumber\\
		=&1-\frac{k}{|Y|}+\frac{k}{|Y|}S_{\rm Cl}^{\max},
	\end{align}
	so that we finally have
	\begin{align}
		S>S_{\rm Cl}^{\max}\Leftrightarrow S_{\rm Bell}>S_{\rm local}^{\max}.
	\end{align}

\section{The one-shot success probability in the Prevedel {\em et al.}~protocol}\label{prevedelsuccessprobability}
The expression for the one-shot success probability in Eq.~\eqref{prevedelexpr} can be obtained as follows: starting from the general expression for $S$, we have
\begin{align}
S=&\sum_mp(m)\sum_xp(x|m)\sum_y\mathcal{N}(y|x)\sum_vp(v|y)\nonumber\\
&\sum_zp(z|v,m,x)p(m'=m|z,y)\\
=&\sum_mp(m)\sum_xp(x|m)\sum_y\mathcal{N}(y|x)\sum_vp(v|y)\nonumber\\
&\sum_zp(z|v,m,x)\delta_{g(z,y),m}.
\end{align}
This then becomes
\begin{align}
S=&\sum_mp(m)\sum_xp(x|m)\mathcal{N}(y=(1,b_1)|x=(b_1,b_2))\nonumber\\
&\sum_vp(v|y)\sum_zp(z|v,m,x)\delta_{b_1,m}\nonumber\\
+&\sum_mp(m)\sum_xp(x|m)\mathcal{N}(y=(2,b_2)|x=(b_1,b_2))\nonumber\\
&\sum_v\delta_{v,1}\sum_zp(z|v,m,x)\delta_{b_2\oplus z,m}\nonumber\\
+&\sum_mp(m)\sum_xp(x|m)\mathcal{N}(y=(P,b_1\oplus b_2)|x=(b_1,b_2))\nonumber\\
&\sum_v\delta_{v,0}\sum_zp(z|v,m,x)\delta_{b_1\oplus b_2\oplus z,m}\nonumber\\
=&\sum_mp(m)\sum_{b_2}p(x=(m,b_2)|m)\nonumber\\
&\mathcal{N}(y=(1,m)|x=(m,b_2))\nonumber\\
+&\sum_mp(m)\sum_xp(x|m)\mathcal{N}(y=(2,b_2)|x)\nonumber\\
&\sum_zp(z|v=1,m,x)\delta_{b_2\oplus z,m}\nonumber\\
+&\sum_mp(m)\sum_xp(x|m)\mathcal{N}(y=(P,b_1\oplus b_2)|x)\nonumber\\
&\sum_zp(z|v=0,m,x)\delta_{b_1\oplus b_2\oplus z,m}.
\end{align}

Finally, for $p(m)=\frac{1}{2}$, we have 
\begin{align}
S=&\sum_m\frac{1}{2}\sum_{b_2}p(x=(m,b_2)|m)\frac{1}{3}\nonumber\\
+&\sum_m\frac{1}{2}\sum_{x}p(x|m)\frac{1}{3}\sum_zp(z|v=1,m,x)\delta_{b_2\oplus z,m}\nonumber\\
+&\sum_m\frac{1}{2}\sum_xp(x|m)\frac{1}{3}\sum_zp(z|v=0,m,x)\delta_{m\oplus b_2\oplus z,m}\nonumber\\
=&\frac{1}{3}+\frac{1}{6}\sum_m\sum_{x}p(x|m)\sum_zp(z|v=1,m,x)\delta_{b_2\oplus z,m}\nonumber\\
+&\frac{1}{6}\sum_m\sum_xp(x|m)\sum_zp(z|v=0,m,x)\delta_{b_2\oplus z,0}\nonumber\\
=&\frac{1}{3}+\frac{1}{6}\sum_m\sum_{x}\sum_zp(x,z|m,v=1)\delta_{b_2\oplus z,m}\nonumber\\
+&\frac{1}{6}\sum_m\sum_x\sum_zp(x,z|m,v=0)\delta_{b_2\oplus z,0}\nonumber\\
=&\frac{1}{3}+\frac{1}{6}\sum_m\sum_{b_2}\sum_zp(b_2,z|m,v=1)\delta_{b_2\oplus z,m}\nonumber\\
+&\frac{1}{6}\sum_m\sum_{b_2}\sum_zp(b_2,z|m,v=0)\delta_{b_2\oplus z,0}\nonumber\\
=&\frac{1}{3}+\frac{1}{6}\sum_{b_2,z,m,v}p(b_2,z|m,v)\delta_{b_2\oplus z,mv}.
\end{align}

\section{The nonlocal game for the  Prevedel {\em et al.}~protocol following Theorem \ref{isomorphism}}\label{prevedelgame}
The Prevedel {\em et al.}~protocol \cite{PLM11} is evidently built around the CHSH game. How does this square with the general construction of a nonlocal game that we referred to in Theorem \ref{isomorphism}? We do a consistency check here. Following the recipe for constructing a nonlocal game starting from a classical channel $\mathcal{N}$, outlined in Section \ref{sec5}, we have the following expression for the probability of success in the nonlocal game:
\begin{align}
	&S_{\rm Bell}\nonumber\\
	=&\frac{1}{12}\sum_{m,m',x,y\notin Y_x}p(x,m'|m,y)\nonumber\\
	+&\frac{1}{12}\sum_{m,m',x,y\in Y_x}p(x,m'|m,y)\delta_{m,m'},
\end{align} 
where we used $p(m,y)=p(m)p(y)=\frac{1}{2}\frac{1}{6}=\frac{1}{12}$.

The expression for the one-shot success probability, $S$, is, of course, given by (using $p(m)=\frac{1}{2}$ and $\mathcal{N}(y|x)=\frac{1}{3}\delta(y\in Y_x)$)
\begin{align}
	S=\frac{1}{6}\sum_{m,m',x,y\in Y_x}p(x,m'|m,y)\delta_{m,m'}.
\end{align}
Hence, using the fact that $k=3$ and $|Y|=6$ in this example, and following Eq.~\eqref{bellpnc}, we have
\begin{align}
	S_{\rm Bell}=\frac{1}{2}+\frac{1}{2}S.
\end{align}
Recalling Eq.~\eqref{prevedelchsh}, 
\begin{align}
	S=\frac{1}{3}+\frac{2}{3}S_{\rm CHSH},
\end{align}
and we therefore have
\begin{align}
	S_{\rm Bell}=\frac{2}{3}+\frac{1}{3}S_{\rm CHSH}
\end{align}
for the nonlocal game defined according to Section \ref{sec5} and used in Theorem \ref{isomorphism}. Hence, we have that the nonlocal game constructed from our general recipe is essentially the CHSH game, except that the two inputs on Bob's side are disguised as six inputs labelled by $y\in Y$ that are classically post-processed to obtain $v\in\{0,1\}$ and the output for each $y$ is given by $m'\in\{0,1\}$ obtained by classically post-processing $z,y$ (Fig.~\ref{schematic}).

\section{Contextuality witnessed by the weighted max-predictability is sufficient for a quantum advantage}\label{hypergraphinvariant}
Recall the expression for $S$ given in Eq.~\eqref{sfncorr}, i.e.,
\begin{equation}
	S={\rm Corr}+\sum_mp(m)\sum_{x,x'\in X^{(m)}}(1-\delta_{x,x'})p(x,x'|m)\eta(x,x'),
\end{equation} 
where, from Eq.~\eqref{defncorr}, we have
\begin{equation}
	{\rm Corr}=\sum_mp(m)\sum_{x,x'\in X^{(m)}}\delta_{x,x'}p(x,x'|m).
\end{equation}
From Eq.~\eqref{sontlmodel}, we have, under the assumption of noncontextuality, that
\begin{align}
	S=&\sum_{\lambda}\nu(\lambda)\sum_mp(m)\nonumber\\
	&\sum_{x,x'\in X^{(m)}}\xi(x'|m,\lambda)\mu(x|m,\lambda)\eta(x,x'),
\end{align}
so that 
\begin{align}
	{\rm Corr(\lambda)}\equiv \sum_mp(m)\sum_{x,x'\in X^{(m)}}\delta_{x,x'}\xi(x'|m,\lambda)\mu(x|m,\lambda),
\end{align}
and
\begin{equation}
	{\rm Corr}=\sum_{\lambda}{\rm Corr}(\lambda)\nu(\lambda)
\end{equation}

We now proceed to upper bound $S^{\max}_{\rm NC}$ in terms of a hypergraph invariant:
\begin{align}\label{smaxnc}
	&S^{\max}_{\rm NC}\nonumber\\
	=&\max_{\lambda}\Bigg({\rm Corr}(\lambda)+\sum_mp(m)\nonumber\\
	&\sum_{x,x'\in X^{(m)}}\Big((1-\delta_{x,x'})\xi(x'|\mathbb{M}^{\rm (A)}_m,\lambda)\nonumber\\
	&\mu(x|\mathbb{S}_m,\lambda)\eta(x,x')\Big)\Bigg)\nonumber\\
	\leq&\max_{\lambda}\left({\rm Corr}(\lambda)+\eta_{\max}(1-{\rm Corr}(\lambda)\right)\nonumber\\
	=&\max_{\lambda}\left(\eta_{\max}+{\rm Corr}(\lambda)(1-\eta_{\max})\right)\nonumber\\
	=&\eta_{\max}+(1-\eta_{\max})\max_{\lambda}{\rm Corr}(\lambda)\nonumber\\
	\leq& \eta_{\max}+(1-\eta_{\max})\beta(\Gamma,\{p(m)\}_m),
\end{align}
where 
\begin{align}
	\beta(\Gamma,\{p(m)\}_m)\equiv\max_{\lambda} \sum_mp(m)\max_{x\in X^{(m)}}\xi(x|m,\lambda)
\end{align}
is the weighted max-predictability \cite{Kunjwal19,Kunjwal20}. Hence, we have 
\begin{align}
	S\leq S^{\max}_{\rm NC}\leq \eta_{\max}+(1-\eta_{\max})\beta(\Gamma,\{p(m)\}_m).
\end{align}

{\bf A special case:} We now consider the special case when $\eta_{\max}=\eta_{\min}=\eta$. Following Eq.~\eqref{sfncorr}, the lower and upper bounds on $S$ coincide and we have 

\begin{align}
	S&={\rm Corr}+\eta (1-{\rm Corr})\\
	&=\eta+{\rm Corr}(1-\eta).\label{etacorrbound}
\end{align}

The upper bound from noncontextuality becomes (cf.~Eq.~\eqref{smaxnc})
\begin{align}
	S^{\max}_{\rm NC}=\eta+\max_{\lambda}{\rm Corr(\lambda)}(1-\eta).
\end{align}

We then have, from noncontextuality, that
\begin{align}
	S&\leq S^{\max}_{\rm NC}\nonumber\\ 
	&\leq \eta+(1-\eta)\beta(\Gamma,\{p(m)\}_m).
\end{align}

Together with Eq.~\eqref{etacorrbound}, this gives us the following:
\begin{align}
	{\rm Corr}&>\beta(\Gamma,\{p(m)\}_m)\nonumber\\
	\Leftrightarrow S&>\eta+(1-\eta)\beta(\Gamma,\{p(m)\}_m)\nonumber\\
	&\geq S^{\max}_{\rm NC}=S^{\max}_{\rm Cl(CIG)}.
\end{align}

Recalling that 
\begin{align}\label{ncineq}
	{\rm Corr}\leq \beta(\Gamma,\{p(m)\}_m)
\end{align}
is a noise-robust noncontextuality inequality \cite{Kunjwal20}, we have that the contextuality witnessed by ${\rm Corr}>\beta(\Gamma,\{p(m)\}_m)$ is sufficient for a quantum advantage when $\eta_{\max}=\eta_{\min}$.

{\bf Another special case:} The sufficiency of ${\rm Corr}>\beta(\Gamma,\{p(m)\}_m)$ for a quantum advantage also arises when $S_{\rm imperf}=0$. Then we have that $S=S_{\rm perf}={\rm Corr}$ and $S^{\rm max}_{\rm NC}\leq \beta(\Gamma,\{p(m)\}_m)$, so that the violation of ${\rm Corr}\leq \beta(\Gamma,\{p(m)\}_m)$ implies the violation of $S\leq S^{\rm max}_{\rm NC}$. Indeed, in the ideal quantum case considered by Cubitt {\em et al.}\cite{CLM10}, we see that ${\rm Corr}=1$, maximally violating the noncontextuality inequality and achieving a success probability of $1$.

\section{Not every KS set admits a KS basis set: the Conway-Kochen $31$-vector KS set}\label{conway31}

Consider the simplest known KS set in $d=3$ dimensions, namely, the Conway-Kochen $31$-vector KS set \cite{Peres06}. The $31$ vectors are carved up into $17$ complete orthogonal bases (with $3$ vectors each) and $20$ incomplete orthogonal bases (with $2$ vectors each). The orthogonality graph has an independence number of $11$ and the only disjoint basis sets of size greater than $11$ are those of size $12$ and $13$. None of these disjoint basis sets forms a KS basis set, hence no quantum advantage over the unassisted one-shot zero-error capacity of $11$ can be obtained via the methods of Ref.~\cite{CLM10} for this  construction. A remaining possibility is that, on adding the missing vectors in the $20$ incomplete orthogonal bases to the KS set, the orthogonality relations between the resulting set of $51$ vectors will perhaps allow for an advantage. We rule out this possibility as well: after including $20$ additional vectors that render all the bases that appear in this KS set complete, we have a contextuality scenario represented by a hypergraph containing $51$ vertices carved up into $37$ (three-vertex) hyperedges. We check for any  additional orthogonality relations arising from the newly introduced $20$ vectors and find $4$ additional incomplete orthogonal bases. On further completing these $4$ bases by adding $4$ more vectors, we find that there are, overall, $55$ vertices (vectors) carved up into $41$ hyperedges (complete orthogonal bases) and there are no additional orthogonality relations.
The orthogonality graph associated with this extended contextuality scenario has an independence number of $25$. Hence, for an advantage based on the strategy of Ref.~\cite{CLM10}, there must exist a KS basis set of size $q>25$. However, the largest disjoint basis set is still of size $13$ and it does not form a KS basis set. Hence, the strategy of Ref.~\cite{CLM10} does not provide an advantage even (and especially) when extending Conway-Kochen 31-vector KS set to complete all incomplete orthogonal bases and include any additional orthogonality relations (leaving no incomplete orthogonal bases). This provides a counter-example to the question we posed, showing that the Cubitt {\em et al.}~strategy doesn't work for arbitrary KS sets.

One might wonder why we bother ``completing" the original $31$-vector KS set to $55$-vector KS set with no incomplete bases. We do this to rule out the possibility that something akin to the $18$-vector KS set in $4$ dimensions \cite{CEGA96} is happening here: for that KS set, it's not possible to implement the Cubitt {\em et al.}~protocol, but supplementing it with the remaining set of $6$ vectors (out of Peres's $24$-vector KS set \cite{Peres91} from which the $18$-vector set is drawn) and taking into account the resulting additional orthogonality relations yields Peres's $24$-vector KS set for which the Cubitt {\em et al.}~protocol works. From our investigation, it is clear that for the $31$-vector KS set, such a situation doesn't arise even after ``completing" it.

We provide below a list of all the vectors and bases in the orginal as well as the ``completed" KS set for the Conway-Kochen argument, so that the interested reader may verify our claims concerning this KS set.

\subsection{Conway-Kochen $31$-vector KS set}
The $31$ vectors (labelled from $1$ to $31$) are:
\begin{align}
&1: (-1,2,1), 2:(-1,2,0), 3: (0,2,1), 4: (-1,2,-1),\nonumber\\
&5: (0,2,0), 6: (1,2,1), 7: (0,2,-1), 8: (1,2,0),\nonumber\\
&9: (1,2,-1), 10: (0,2,-2), 11: (2,2,0), 12: (2,2,-2),\nonumber\\
&13: (-1,1,-2), 14: (0,1,-2), 15: (-1,0,-2),\nonumber\\
&16: (0,0,-2), 17: (-1,-1,-2), 18: (0,-1,-2),\nonumber\\
&19: (0,-2,-2), 20: (2,1,-1), 21: (2,1,0), \nonumber\\
&22: (2,0,-1), 23: (2,0,0), 24: (2,-1,-1), \nonumber\\
&25: (2,-1,0), 26: (2,-2,0), 27: (2,0,-2), \nonumber\\
&28: (2,-2,-2), 29: (2,2,2), 30: (2,0,2),\nonumber\\
&31: (2,-2,2).
\end{align}

The orthogonality relations of between these vectors are the following (the first entry in each list is the vector with respect to which the remaining vectors in the list are orthogonal):

\begin{align}
&[1; 12, 14, 21, 30], [2; 16, 20, 21], [3; 13, 14, 23], \nonumber\\
&[4; 18, 21, 27,
29],[5; 15, 16, 22, 23, 27, 30],\nonumber\\
&[6; 14, 25, 27, 31], [7; 17, 18, 23], [8; 16, 24, 25],\nonumber\\
&[9; 18, 25, 28, 30], [10; 19, 23, 24, 28, 29],\nonumber\\
&[11; 13,
16, 26, 28, 31], [12; 1, 17, 19, 26, 30],\nonumber\\
&[13; 3, 11, 22, 28], [14; 1,
3, 6, 23], [15; 5, 20, 22, 24],\nonumber\\
&[16; 2, 5, 8, 11, 21, 23, 25, 26], [17;
7, 12, 22, 26],\nonumber\\
&[18; 4, 7, 9, 23], [19; 10, 12, 20, 23, 31], [20; 2, 15,
19, 31],\nonumber\\
&[21; 1, 2, 4, 16], [22; 5, 13, 15, 17],\nonumber\\
&[23; 3, 5, 7, 10, 14,
16, 18, 19], [24; 8, 10, 15, 29],\nonumber\\
&[25; 6, 8, 9, 16], [26; 11, 12, 16,
17, 29],\nonumber\\
&[27; 4, 5, 6, 29, 30, 31], [28; 9, 10, 11, 13, 30],\nonumber\\
&[29; 4, 10,
24, 26, 27], [30; 1, 5, 9, 12, 27, 28],\nonumber\\
&[31; 6, 11, 19, 20, 27].
\end{align}

The $17$ complete orthogonal bases (vectors labelled as above) are:
\begin{align}
&\{1, 12, 30\}, \{2, 16, 21\}, \{3, 14, 23\}, \{4, 27, 29\}, \{5, 15, 22\},\nonumber\\
&\{5, 16, 23\},\{5, 27, 30\}, \{6, 27, 31\}, \{7, 18, 23\}, \{8, 16, 25\},\nonumber\\
&\{9, 28, 30\}, \{10, 19, 23\}, \{10, 24, 29\}, \{11, 13, 28\},\nonumber\\
&\{11, 16, 26\},\{12, 17, 26\}, \{19, 20, 31\}.
\end{align}

The $20$ incomplete orthogonal bases are:
\begin{align}
&\{1, 14\}, \{1, 21\}, \{2, 20\}, \{3, 13\}, \{4, 18\},\nonumber\\
&\{4, 21\}, \{6, 14\}, \{6, 25\}, \{7, 17\}, \{8, 24\},\nonumber\\
&\{9, 18\}, \{9, 25\}, \{10, 28\}, \{11, 31\}, \{12, 19\},\nonumber\\
&\{13, 22\}, \{15, 20\}, \{15, 24\}, \{17, 22\}, \{26, 29\}.
\end{align}

\subsection{The ``completed" $55$-vector KS set}

The $20$ vectors that complete the incomplete bases are: 
\begin{align}
&32: (4, -4, -8), 33: (-5, -1, 2), 34: (-4,-4,8),\nonumber\\ 
&35: (-1, -2, -5), 36: (5, 2, -1), 37: (-1, -5, -2),\nonumber\\
&38: (1, -5, 2), 39: (-1, -2, 5), 40: (-8, -4, -4),\nonumber\\
&41: (5, -1, -2), 42: (-5, 2, -1), 43: (2, 5, -1),\nonumber\\
&44: (-2, 1, -5), 45: (1, -2, -5), 46: (-2, -1, -5),\nonumber\\
&47: (-5, -2, -1), 48: (-5, 2, 1), 49: (8, -4, 4),\nonumber\\
&50: (-1, 2, -5), 51: (-2, 5, 1).
\end{align}

The resulting set of $17$ newly complete bases is then:
\begin{align}
&\{1, 14, 47\}, \{1, 21, 50\}, \{2, 20, 46\}, \{3, 13, 33\}, \{4, 18, 36\},\nonumber\\
&\{4, 21, 45\}, \{6, 14, 48\}, \{6, 25, 39\}, \{7, 17, 41\}, \{8, 24, 44\},\nonumber\\
&\{9, 18, 42\}, \{9, 25, 35\}, \{10, 28, 40\}, \{11, 31, 32\}, \nonumber\\
&\{12, 19, 49\}, \{13, 22, 37\}, \{15, 20, 51\}, \{15, 24, 43\},\nonumber\\
&\{17, 22, 38\}, \{26, 29, 34\}.
\end{align}

Taking into account possible extra orthogonality relations not captured by the set of $37$ complete bases, it turns out that there are $4$ additional incomplete bases in the set of $51$ vectors above:
\begin{align}
\{2, 40\}, \{3, 34\}, \{7, 32\}, \{8, 49\}.
\end{align}
To complete these bases we add $4$ more vectors to the $51$-vector KS set:
\begin{align}
&52: (20, 4, 8), 53: (-20, 4, -8),\nonumber\\
&54: (8, 4, -20), 55: (8, -4, -20).
\end{align}
The $4$ additional newly complete bases are then
\begin{align}
\{2, 40, 54\}, \{3, 34, 53\},  \{7, 32, 52\}, \{8, 49, 55\}.
\end{align}
There are no new orthogonality relations in this completed set of $55$ vectors carved up into $41$ complete orthogonal bases.
\end{appendix}

\end{document}